%% file: main.tex
\documentclass[11pt,reqno]{amsart}

\usepackage[T1]{fontenc}
\usepackage{lmodern}
\usepackage{microtype}
\usepackage{amsmath,amssymb,mathtools,bm,mathrsfs}
\usepackage{booktabs,array,tabularx}
\usepackage{graphicx}
\usepackage{placeins}
\usepackage{flafter}
\usepackage{needspace}
\usepackage{enumitem}
\usepackage{geometry}
\usepackage{etoolbox,xspace}
\usepackage{xcolor}
\usepackage[breaklinks,colorlinks]{hyperref}
\usepackage[nameinlink,noabbrev]{cleveref}

\definecolor{DarkPaper}{HTML}{111318}
\definecolor{DarkText}{HTML}{E8EAF0}
\definecolor{DarkLink}{HTML}{78C7FF}
\definecolor{DarkCite}{HTML}{9FE6C3}
\definecolor{DarkURL}{HTML}{D5B8FF}
\pagecolor{DarkPaper}
\color{DarkText}
\hypersetup{
  pdftitle={Benign projective landscapes for measured quantum divergences},
  pdfauthor={Domingos S. P. Salazar},
  pdfsubject={Measured quantum divergences, accessible information, adaptive decoding, and benign nonconvexity},
  pdfkeywords={no spurious local maxima, strict saddle, benign nonconvexity, landscape analysis, measured quantum divergence, accessible information, adaptive decoding, measurement order},
  linkcolor=DarkLink,
  citecolor=DarkCite,
  urlcolor=DarkURL
}
\graphicspath{{figures/}}
\allowdisplaybreaks[2]
\newtheorem{theorem}{Theorem}[section]
\newtheorem{proposition}[theorem]{Proposition}
\newtheorem{lemma}[theorem]{Lemma}
\newtheorem{corollary}[theorem]{Corollary}

\newtheorem{question}[theorem]{Question}
\newtheorem{algorithm}[theorem]{Algorithm}
\newtheorem*{informaltheorem}{Main theorem (informal)}
\theoremstyle{definition}
\newtheorem{definition}[theorem]{Definition}
\newtheorem{example}[theorem]{Example}

\theoremstyle{remark}
\newtheorem{remark}[theorem]{Remark}

\DeclareMathOperator{\Tr}{Tr}
\DeclareMathOperator{\spec}{spec}
\DeclareMathOperator{\supp}{supp}

\DeclareMathOperator{\spanop}{span}
\newcommand{\cH}{\mathcal H}
\newcommand{\cD}{\mathcal D}
\newcommand{\cP}{\mathcal P}
\newcommand{\cL}{\mathcal L}
\newcommand{\cF}{\mathfrak F}
\newcommand{\Id}{\mathbf 1}
\newcommand{\RR}{\mathbb R}
\newcommand{\CC}{\mathbb C}
\newcommand{\M}{\mathsf M}
\newcommand{\POVM}{\mathsf{POVM}}

\newcommand{\OFL}{\mathsf{OFL}}
\newcommand{\dd}{\,\mathrm d}
\newcommand{\eps}{\varepsilon}
\newcommand{\ket}[1]{\lvert #1\rangle}
\newcommand{\bra}[1]{\langle #1\rvert}
\newcommand{\proj}[1]{\lvert #1\rangle\!\langle #1\rvert}
\newcolumntype{Y}{>{\raggedright\arraybackslash}X}

\title[Benign projective landscapes for measured quantum divergences]{Benign Projective Landscapes for Measured Quantum Divergences}
\author{Domingos S. P. Salazar}
\address{Unidade de Educa\c{c}\~ao a Dist\^ancia e Tecnologia, Universidade Federal Rural de Pernambuco, 52171-900 Recife, Pernambuco, Brazil}
\date{August 16, 2026}

\begin{document}
\begin{abstract}
We study nonconvex optimization of measured quantum $f$-divergences over
rank-one projective measurements.  For smoothly operator-Fenchel liftable
generators and faithful states, every projective local maximum and every
second-order stationary point is globally optimal over all POVMs.  The
criterion is blockwise: a critical PVM is optimal exactly when the compressed
states are proportional on each equal-score block; otherwise an explicit
two-vector rotation has positive ascent curvature.  Operator-convex generators
admit a positive atomic curvature resolution, and the quadratic $\chi^2$ case
yields two-sided residual bounds.  For binary accessible information, this
framework proves the known adaptive-capacity equality and shows that every
nonidentical qubit ensemble has exactly two stationary projective measurements,
proving conjectures of Keil and Thai--Dall'Arno.  A rare-prior limit connects
weighted Jensen--Shannon information to relative entropy and yields a finite
counterexample to the proposed equivalence between observational-entropy and
all-ensemble mutual-information orders.  The landscape theorem also covers
measured R\'enyi divergences of finite positive order and measured relative
entropy.
\end{abstract}

\maketitle

\begin{center}
\begin{minipage}{0.94\textwidth}
\small\textbf{AI-use disclosure.}
GPT 5.6 Sol was used to assist with literature retrieval, exact symbolic
arithmetic checks, and preparation of an initial draft. Responsibility for
the mathematical statements and the final manuscript rests with the author.
\par\bigskip
\textbf{Outline.}\par\smallskip
\emph{Foundations:}
\hyperref[sec:introduction]{motivation and main result},
\hyperref[sec:preliminaries]{measured divergences and the operator lift},
\hyperref[sec:lift]{score operators}.
\par
\emph{Core landscape theory:}
\hyperref[thm:benign-score-block]{the benign score-block theorem},
\hyperref[sec:atomic]{positive atomic curvature},
\hyperref[thm:quadratic-robust]{quadratic error bounds},
\hyperref[thm:renyi-landscape]{R\'enyi landscapes}.
\par
\emph{Binary information:}
\hyperref[sec:shor]{projective sufficiency},
\hyperref[thm:binary-landscape]{the accessible-information landscape},
\hyperref[thm:keil]{the qubit classification},
\hyperref[thm:thai-conjectures]{quasi-concavity and bisection}.
\par
\emph{Consequences and outlook:}
\hyperref[thm:adaptive-capacity]{adaptive capacity},
\hyperref[thm:measurement-orders]{measurement orders},
\hyperref[sec:outlook]{limitations and open questions},
\hyperref[app:calculus]{appendices}.
\end{minipage}
\end{center}

\input{sections/01-introduction}
\input{sections/02-preliminaries}
\input{sections/03-score-lift}
\input{sections/04-landscape}
\input{sections/05-atomic}
\input{sections/09-renyi}
\input{sections/03-shor}
\input{sections/06-accessible}
\input{sections/07-qubit}
\input{sections/08-thai}
\input{sections/07-consequences}
\input{sections/10-outlook}
\appendix
\input{appendices/A-calculus}
\input{appendices/B-triangular}
\input{appendices/C-boundary}
\input{appendices/D-literature}

\FloatBarrier
\bibliographystyle{amsplain-links}
\bibliography{references}
\end{document}

%% file: sections/01-introduction.tex
\section{Introduction}\label{sec:introduction}

\begin{figure}[!t]
 \centering
 \includegraphics[width=0.72\textwidth]{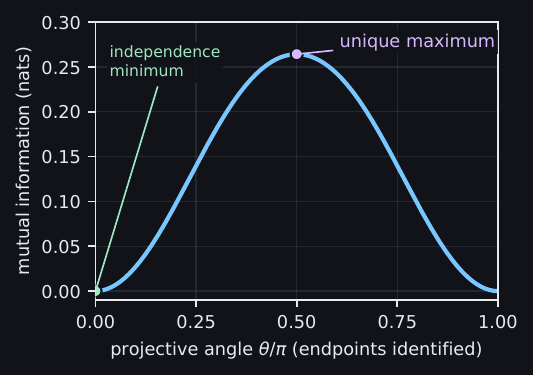}
 \caption{Mutual information for equal-prior faithful qubits with Bloch
 vectors $r_{0,1}=(\pm2\sqrt3/5,0,2/5)$, plotted over the projective circle.
 Antipodal axes represent the same PVM.  The marked independence minimum
 and information-optimal maximum are the only stationary PVMs.}
 \label{fig:qubit-landscape}
\end{figure}

For a binary qubit ensemble, the projective measurements that matter lie on
a circle: the measurement axis runs in the Bloch plane of the two states,
with antipodal axes identified.  Mutual information along this circle is
continuous and nonconvex, and it is smooth away from singular axes.  The
typical picture is nevertheless as
simple as possible---one independence minimum and one information-optimal
maximum, with no other stationary measurement; see Fig.~\ref{fig:qubit-landscape}.
We prove that this picture holds for every nonidentical binary qubit
ensemble, including singular boundary cases.  This proves the
two-stationary-point conjecture posed as Conjecture~2 on p.~77 of Keil's
2009 thesis \cite{KeilThesis2009}.  Keil's 2008 preprint, published in
2024, had already proved projective sufficiency for binary qubit ensembles
and the corresponding local-maximum structure \cite{Keil2024}.  Through
the coordinate introduced by Thai and Dall'Arno, the result also proves
their Conjectures~1 and 2, proposed in their December 2025 preprint
\cite{ThaiDallArno2026}.

Accessible-information optimization has a long history, including general
outcome bounds, binary-state analyses, ensemble-dependent upper and lower
bounds, exact symmetric-source results, and numerical receiver methods
\cite{Davies1978,Levitin1995,Fuchs1996,FuchsCaves1994,
JozsaRobbWootters1994,SasakiBarnettJozsaOsakiHirota1999,
RehacekEnglertKaszlikowski2005}.  These results supply important global and
first-order information; the issue here is the full critical-point geometry
of the projective objective.

The qubit result is the two-dimensional shadow of a general landscape
theorem.  Let $\rho$ and $\sigma$ be density operators and let $f$ be a
strictly convex generator.  A rank-one projective measurement
$P=(P_j)_{j=1}^d$ produces
\begin{equation}\label{eq:intro-probabilities}
 p_j=\Tr(\rho P_j),\qquad q_j=\Tr(\sigma P_j),\qquad t_j=p_j/q_j,
\end{equation}
and the objective $F_f(P)=\sum_jq_jf(t_j)$.  Even when a projective
measurement is known to attain the full POVM optimum, optimization over
orthonormal bases remains nonconvex.  Our result identifies exactly why
this restricted nonconvexity has no spurious local maxima.

The binary lift also revisits two questions outside static landscape
geometry, with different priority status.  Shor's Conjecture~2, posed in
the 2002 preprint and published in 2004, asks whether adaptive
single-system decoding can exceed separate-measurement capacity for a
channel with two mixed outputs \cite{ShorAdaptive2004}.  Xiao's August
2021 public report states and presents an argument for both concavity of
binary accessible information in the prior and the equality
$C_{1,A}=C_{1,1}$ \cite{Xiao2021}.  We give a different
operator-perspective proof and an explicit transcript-potential reduction,
with no priority claim for the equality.  Separately, Teixid\'o-Bonfill,
Schindler, and \v{S}afr\'anek conjectured in their 2023 preprint that
all-ensemble mutual-information order is equivalent to
observational-entropy order \cite{TeixidoSchindlerSafranek2025}.  Combining
the known less-noisy/relative-entropy equivalence with their published
Example~16, our rare-prior argument gives an explicit finite ensemble
witness against that conjectured equivalence.

\Needspace{0.48\textheight}
\begin{informaltheorem}
Let $f\in\cF_{\OFL}$ and let $\rho,\sigma\in\cD_d^\circ$.  For a critical
rank-one PVM $P=(P_j)$, form the score operator and its distinct-score blocks
by
\begin{equation}\label{eq:intro-score}
 \Gamma_P=\sum_j\vartheta(t_j)P_j
 =\sum_{\alpha=1}^m s_\alpha B_\alpha,
 \qquad t_\alpha=\vartheta^{-1}(s_\alpha).
\end{equation}
Then $P$ is globally optimal over all POVMs if and only if
\begin{equation}\label{eq:intro-rigidity}
 B_\alpha(\rho-t_\alpha\sigma)B_\alpha=0
 \qquad\text{for every }\alpha.
\end{equation}
If this condition fails, some $B_\alpha$ contains basis vectors $j\ne k$
for which a phase-adjusted two-level rotation satisfies
\begin{align}
 \frac{\dd^2}{\dd\theta^2}F_f(P(\theta))\Big|_{\theta=0}
 &=4f''(t_\alpha)
 \left|\bra j(\rho-t_\alpha\sigma)\ket k\right|^2\notag\\
 &\quad\times\left(\frac1{q_j}+\frac1{q_k}\right)>0.
 \label{eq:intro-curvature}
\end{align}
Consequently every projective local maximum and every second-order
stationary PVM is globally optimal over all POVMs, and every nonglobal
critical PVM has an explicit positive-curvature pair rotation.
\end{informaltheorem}

\noindent The exact statement and proof are given in
Theorem~\ref{thm:benign-score-block}.

\subsection{Lift, block, and atom}

The proof has three layers.  First, the scalar Fenchel optimizer for each
outcome is assembled into the score operator in \eqref{eq:intro-score}.
For the smoothly operator-Fenchel liftable generators introduced below,
Fang, Fawzi, and Fawzi's operator-Jensen theorem embeds the full measured
divergence into a concave optimization over one Hermitian operator
\cite{FangFawziFawzi2026}.

Second, at a critical PVM the equal eigenvalues of $\Gamma_P$ define score
blocks $B_\alpha$.  Inside a block, all outcomes have one likelihood ratio
$t_\alpha$.  We prove that the critical PVM is globally optimal precisely
when \eqref{eq:intro-rigidity} holds.
If this proportionality fails, an off-diagonal element of
$B_\alpha(\rho-t_\alpha\sigma)B_\alpha$ selects two basis vectors and an
explicit pair rotation with positive second variation.

Third, for operator-convex $f$, the scalar curvature is a positive
superposition of the resolvent-atom curvatures
\begin{equation}\label{eq:intro-atom}
 h_\lambda(t)=\frac{(t-1)^2}{(1-\lambda)t+\lambda}.
\end{equation}
The classical operator-convex integral representation
\cite{HansenPedersen1982,BhatiaMatrix1997}, in the resolvent-atom form used
by Salazar for the positive decomposition of Petz quantum $f$-divergences
\cite{SalazarAtomic2026}, permits the escape curvature in the present
projective problem to be resolved atom by atom: every atom detects the same
pair rotation with the same sign.

This theorem turns an existence statement into a local-to-global criterion.
For the quadratic atom $h_1(t)=(t-1)^2$, completing the matrix square makes
the criterion quantitative: Theorem~\ref{thm:quadratic-robust} gives
explicit two-sided bounds on both the global value gap and the distance to
the unique optimal score operator in terms of the commutator, the score-block
residuals, and the minimum score gap.  A separated-score corollary controls
the distance to an optimal PVM.
The binary assertion now commonly called Shor's orthogonal-measurement
conjecture asks whether some optimal von Neumann measurement exists for
every binary ensemble.  Shor's 2000 report records the underlying
Fuchs--Peres numerical evidence and discusses Levitin's broader conjecture,
without formulating this surviving binary statement as a numbered
conjecture \cite{Levitin1995,Shor2000}.  Wang, Wang, and Chen explicitly noted that the
2025 Fang--Fawzi--Fawzi theorem already implies binary projective
sufficiency \cite{WangWangChen2026};
Section~\ref{sec:shor} records the short weighted Jensen--Shannon
specialization.  Theorem~\ref{thm:benign-score-block} adds the critical-point
geometry: any local
search that reaches second-order stationarity has found the global POVM
optimum, and every critical point where it has not comes with a computable
escape rotation.

\begin{table*}[!t]
 \caption{Notation used throughout the paper.}
 \label{tab:notation}
 \centering
 \setlength{\fboxsep}{6pt}
 \fbox{\begin{minipage}{0.96\textwidth}
 \small
 \begin{tabularx}{\linewidth}{@{}lY@{}}
 \toprule
 Symbol & Meaning \\
 \midrule
 $\rho,\sigma$ & Quantum states defining the measured divergence. \\
 $P=(P_j)$ & Ordered rank-one PVM, a point of $U(d)/U(1)^d$. \\
 $p_j,q_j,t_j$ & Outcome probabilities and likelihood ratio in
 \eqref{eq:intro-probabilities}. \\
 $F_f$ & Projective $f$-divergence objective. \\
 $\psi,g=f^*\!\circ\psi,\vartheta$ & Operator coordinate, convex companion,
 and scalar score map. \\
 $\Gamma_P$ & Hermitian score operator associated with $P$. \\
 $\cL_f,G_P$ & Concave matrix lift and its gradient at $\Gamma_P$. \\
 $s_\alpha,B_\alpha,t_\alpha$ & Distinct score, its spectral block, and the
 corresponding likelihood ratio. \\
 $A_\alpha$ & Block residual
 $B_\alpha(\rho-t_\alpha\sigma)B_\alpha$. \\
 $R_P,\Phi_{\mathcal E}$ & Posterior effect and its concave binary-information
 functional. \\
 $\nu_f,w_{a,b}$ & Positive atomic measure and the triangular
 Jensen--Shannon density. \\
 \bottomrule
 \end{tabularx}
 \end{minipage}}
\end{table*}

\subsection{Applications and attribution}

For binary accessible information the score is the posterior probability,
and the operator-convex atomic density is an explicit triangular Green
kernel.  Joint concavity of the posterior program in the prior and effect
gives a different operator-perspective proof of prior-value concavity.  A
transcript potential then gives a proof in the present framework of Shor's
2002/2004 mixed-state adaptive-capacity equality $C_{1,A}=C_{1,1}$
\cite{ShorAdaptive2004}.  Xiao's 2021 public report already states and
presents an argument for both prior concavity and this equality
\cite{Xiao2021}; we make no first-proof claim.  The contribution here is
the joint operator perspective and the explicit transcript-potential
reduction.

The classical less-noisy characterization identifies all-ensemble
mutual-information order with pointwise relative-entropy order
\cite{KornerMarton1977,MakurPolyanskiy2018}.  Teixid\'o-Bonfill, Schindler,
and \v{S}afr\'anek conjectured in 2023 that this order is equivalent to
observational-entropy order, while their Example~16 already supplied a
qubit pair with observational-entropy dominance and failure of the
corresponding relative-entropy order \cite{TeixidoSchindlerSafranek2025}.
The present boundary argument writes a concrete finite rare-prior ensemble
witnessing the mutual-information reversal and proves that the complete
normalized triangular density converges in $L^1$ to the relative-entropy
atomic density.  No priority is claimed for the less-noisy equivalence or
for the published measurement-pair separation.

For qubits, score-block rigidity leaves only the independence minimum and
the unique global maximum.  The same general theorem gives benign
projective landscapes for measured R\'enyi divergences of every finite positive
order and for measured relative entropy.

Priority and mathematical validity are kept separate.  The imported
matrix lift, the positive atomization, and projective sufficiency are prior;
the results established here are the score-block local-to-global theorem, its
explicit Hessian and atomic resolution, the quadratic residual error bound,
and the qubit and R\'enyi landscape consequences.  Appendix~\ref{app:literature} compares these claims with prior
work, including the precise overlap with the Wang--Wang--Chen preprint.

\paragraph{Logical roadmap.}
The formal result graph has one imported root.  The operator-Fenchel
projectivization theorem feeds the exact score lift; its envelope identity
gives the commutator first variation, score-block rigidity, and the explicit
pair rotation that together prove Theorem~\ref{thm:benign-score-block}.
Four branches then separate cleanly.  Positive atomization gives the
curvature resolution and quadratic stability bounds.  Power generators give
the measured R\'enyi and relative-entropy landscapes.  The weighted
Jensen--Shannon generator gives the binary landscape, the complete qubit
classification, and the Thai--Dall'Arno consequences.  Joint prior--effect
concavity and the rare-prior limit give the adaptive and measurement-order
results.  Thus every application points back either to the flagship
critical-point theorem or to a precisely identified feature of the same
posterior lift.

\subsection{Organization}

Section~\ref{sec:preliminaries} fixes notation and states the imported
projectivization theorem.  Sections~\ref{sec:lift}--\ref{sec:landscape}
develop the score calculus, prove the benign score-block theorem, and give
the exact optimality test and a worked three-dimensional saddle.
Section~\ref{sec:atomic} gives the positive atomic resolution and quadratic
stability bounds, followed in Section~\ref{sec:renyi} by the broad measured
R\'enyi and relative-entropy consequences.  Section~\ref{sec:shor} records
the prior binary projective-sufficiency implication, and
Section~\ref{sec:accessible} develops the binary landscape and triangular
kernel.  Sections~\ref{sec:qubit}--\ref{sec:thai} prove the qubit conjectures.
Section~\ref{sec:lift-consequences} gives the posterior/adaptive and
rare-prior measurement-order branches.  Section~\ref{sec:outlook} states
the limitations and open questions.  The appendices collect
matrix calculus, the triangular density, singular qubit analysis, and the
comparison with prior work.

%% file: sections/02-preliminaries.tex
\section{Measured divergences, projective measurements, and operator lifts}\label{sec:preliminaries}

\subsection{States and measurements}

Throughout, $\cH\cong\CC^d$ is finite dimensional and $\M_d^{\mathrm{sa}}$ denotes the real vector space of Hermitian matrices. The density operators are
\begin{equation}
 \begin{aligned}
 \cD_d&=\{\rho\in\M_d^{\mathrm{sa}}:\rho\ge0,\ \Tr\rho=1\},\\
 \cD_d^\circ&=\{\rho\in\cD_d:\rho>0\}.
 \end{aligned}
\end{equation}
A finite POVM is a family $M=(M_y)_{y\in\mathcal Y}$ with $M_y\ge0$ and $\sum_yM_y=\Id$. A rank-one PVM is an ordered family
\begin{equation}
 P=(P_1,\ldots,P_d),\qquad P_j=\proj{e_j},
\end{equation}
for an orthonormal basis $(e_j)_{j=1}^d$. The space of ordered rank-one PVMs is the complete flag manifold
\begin{equation}
 \cP_d\cong U(d)/U(1)^d.
\end{equation}
All objectives below are invariant under permutations of outcomes, so ordering is only a convenient smooth cover.
We use the standard quotient-manifold geometry; background on orthogonality
constraints and optimization on matrix manifolds can be found in
\cite{EdelmanAriasSmith1998,AbsilMahonySepulchre2008,Boumal2023}.

A tangent path can be written
\begin{equation}
 P_j(\eps)=e^{i\eps K}P_je^{-i\eps K},\qquad K=K^*.
\end{equation}
Its derivative is $\dot P_j=i[K,P_j]$. Diagonal entries of $K$ generate only phase rotations and therefore represent the zero tangent direction.

\subsection{Classical and measured \texorpdfstring{$f$}{f}-divergences}

Following the classical divergence framework of Ali--Silvey and Csisz\'ar
\cite{AliSilvey1966,Csiszar1967}, let
$f:(0,\infty)\to\RR$ be convex. We use the convention
\begin{equation}\label{eq:classical-f}
 D_f(p\Vert q)=\sum_y q_yf\!\left(\frac{p_y}{q_y}\right)
\end{equation}
when $q_y>0$, with the usual lower-semicontinuous boundary extension otherwise. We normalize $f(1)=0$. Adding a multiple of $t-1$ does not change \eqref{eq:classical-f} for probability distributions.

For $\rho,\sigma\in\cD_d$, define
\begin{align}
 D_f^{\mathbb M}(\rho\Vert\sigma)
 &=\sup_{M\in\POVM}D_f(p^M_\rho\Vert p^M_\sigma),\label{eq:measured-f}\\
 D_f^{\mathrm{proj}}(\rho\Vert\sigma)
 &=\sup_{P\in\cP_d}D_f(p^P_\rho\Vert p^P_\sigma).\label{eq:projective-f}
\end{align}
For a fixed $P\in\cP_d$ and faithful states, put
\begin{equation}\label{eq:pqt}
 p_j=\Tr(\rho P_j),\qquad q_j=\Tr(\sigma P_j),\qquad t_j=\frac{p_j}{q_j},
\end{equation}
and
\begin{equation}\label{eq:Ff}
 F_f(P)=\sum_{j=1}^dq_jf(t_j).
\end{equation}

\begin{lemma}[Rank-one refinement]\label{lem:rank-one-refinement}
For convex $f$, refining a projective measurement cannot decrease the classical $f$-divergence. Consequently the supremum over all finite PVMs equals the supremum over rank-one PVMs.
\end{lemma}

\begin{proof}
A coarse outcome is obtained from its refined outcomes by a stochastic map. The classical data-processing inequality for $f$-divergences, equivalently the Jensen argument applied within each coarse cell, gives the claim.
\end{proof}

\subsection{Fenchel duality and matrix functional calculus}

The Fenchel conjugate of the chosen global extension of $f$ is
\begin{equation}
 f^*(z)=\sup_{t\in\RR}\{tz-f(t)\}.
\end{equation}
We use standard finite-dimensional convex-duality conventions
\cite{Rockafellar1970}.
When $f\in C^2(0,\infty)$ and $f''>0$, the map $f'$ is strictly increasing, and on its range
\begin{equation}\label{eq:fenchel-equality}
 f(t)=t f'(t)-f^*(f'(t)),\qquad (f^*)'(f'(t))=t.
\end{equation}

For a $C^1$ scalar function $h$ on an interval and a Hermitian matrix $A$ with spectrum in that interval, $Dh_A$ denotes the Fr\'echet derivative of the functional calculus. If
\begin{equation}
 A=\sum_\alpha a_\alpha E_\alpha
\end{equation}
uses the distinct eigenvalues, then
\begin{equation}\label{eq:divided-difference}
 Dh_A(X)=\sum_{\alpha,\beta}h^{[1]}(a_\alpha,a_\beta)E_\alpha X E_\beta,
\end{equation}
where
\begin{equation}
 h^{[1]}(x,y)=
 \begin{cases}
 \dfrac{h(x)-h(y)}{x-y},&x\ne y,\\[1ex]
 h'(x),&x=y.
 \end{cases}
\end{equation}
The map $Dh_A$ is self-adjoint for the Hilbert--Schmidt inner product. Two consequences used repeatedly are
\begin{align}
 E_\alpha Dh_A(X)E_\alpha&=h'(a_\alpha)E_\alpha X E_\alpha,\label{eq:block-derivative}\\
 [A,Dh_A(X)]&=[h(A),X].\label{eq:commutator-derivative}
\end{align}
They are immediate from \eqref{eq:divided-difference}; a proof is recalled in Appendix~\ref{app:calculus}.
For the matrix functional calculus and operator-convexity background, see
\cite{BhatiaMatrix1997,HansenPedersen1982,Effros2009}.

\subsection{Operator-Fenchel liftability}

\begin{definition}[Operator-Fenchel liftable generator]\label{def:ofl}
A proper lower-semicontinuous convex function
$f:\RR\to\RR\cup\{+\infty\}$ with $(0,\infty)\subset\operatorname{dom}f$
is called \emph{smoothly operator-Fenchel liftable}, abbreviated
$f\in\cF_{\OFL}$, if the following hold.
\begin{enumerate}
\item $f(1)=0$, $f\in C^2(0,\infty)$, and $f''(t)>0$ for every $t>0$.
\item The interior $I_f$ of $\operatorname{dom}f^*$ is the range of $f'$.
\item There are an open interval $J$ and a $C^2$ diffeomorphism
\begin{equation}
 \psi:J\longrightarrow I_f
\end{equation}
with $\psi'(s)\ne0$ for all $s\in J$, such that $\psi$ is operator concave and
\begin{equation}
 g:=f^*\circ\psi
\end{equation}
is operator convex on $J$.
\end{enumerate}
The associated \emph{score map} is
\begin{equation}\label{eq:score-map}
 \vartheta(t)=\psi^{-1}(f'(t)),\qquad t>0.
\end{equation}
\end{definition}

The hypotheses are a smooth finite-dimensional interior specialization of
the sufficient condition in \cite[Theorem 2]{FangFawziFawzi2026}.\footnote{The
printed theorem of Fang--Fawzi--Fawzi assumes a one-to-one operator-concave
map $\psi:J\to\operatorname{dom}f^*$ for a proper lower-semicontinuous
convex $f$, with $f^*\!\circ\psi$ operator convex.  Definition~\ref{def:ofl}
adds strict $C^2$ regularity and a diffeomorphic parametrization of
$I_f=\operatorname{int}\operatorname{dom}f^*$; it is therefore not a
verbatim restatement.  For faithful finite-dimensional pairs, all outcome
likelihood ratios lie in a compact subinterval of $(0,\infty)$, so the
operator-Jensen proof is confined to this smooth interior.}
Strict convexity of $f$ and injectivity of $\psi$ imply that $\vartheta$ is
injective.

For $\Gamma\in\M_d^{\mathrm{sa}}$ with $\spec\Gamma\subset J$, define
\begin{equation}\label{eq:lift-functional}
 \cL_f^{\rho,\sigma}(\Gamma)
 =\Tr\rho\,\psi(\Gamma)-\Tr\sigma\,g(\Gamma).
\end{equation}
Operator concavity of $\psi$ and operator convexity of $g$ imply that $\cL_f^{\rho,\sigma}$ is concave on the convex spectral domain
\begin{equation}
 \Omega_J=\{\Gamma\in\M_d^{\mathrm{sa}}:\spec\Gamma\subset J\}.
\end{equation}
\begin{theorem}[Imported operator-Fenchel projectivization theorem]\label{thm:fff}
Let $f$ satisfy Definition~\ref{def:ofl}.  Then, for faithful
finite-dimensional states $\rho,\sigma$,
\begin{equation}\label{eq:fff-equality}
 D_f^{\mathbb M}(\rho\Vert\sigma)
 =D_f^{\mathrm{proj}}(\rho\Vert\sigma)
 =\sup_{\Gamma\in\Omega_J}\cL_f^{\rho,\sigma}(\Gamma).
\end{equation}
\end{theorem}

\begin{proof}[Source and scope]
This is the faithful finite-dimensional smooth-interior specialization of
\cite[Theorem~2]{FangFawziFawzi2026}.  Fang, Fawzi, and Fawzi identify their
projective variational formula with \cite[Theorem~5.7]{HiaiBook2021} and
describe their projectivization criterion as similar to
\cite[Theorem~5.8]{HiaiBook2021}; we use their finite-dimensional
operator-Jensen formulation.  Spectral decomposition and scalar Fenchel
duality give the projective variational equality, and operator Jensen
compares arbitrary POVMs.  The measured R\'enyi and relative-entropy
formulas used in the principal examples were already established by Berta,
Fawzi, and Tomamichel \cite{BertaFawziTomamichel2017}.  The broader
measured-divergence framework and its relation to other quantum
$f$-divergences are developed in
\cite{Petz1986,HiaiMosonyi2017,HiaiBook2021}.  We use all of these
variational ingredients as prior results.  Boundary-state applications
below supply their own regularization argument; the smooth statement alone
does not cover the boundary.
\end{proof}

\begin{remark}[Boundary conventions]
The landscape theorem in Sections~\ref{sec:lift}--\ref{sec:landscape} is stated for faithful states so that every $p_j,q_j$ is positive and the projective objective is smooth on the whole flag manifold. The local proof only needs positivity of the probabilities at the measurement under consideration. We exploit that observation in the qubit boundary analysis of Section~\ref{sec:qubit} and Appendix~\ref{app:boundary}.
\end{remark}

%% file: sections/03-score-lift.tex
\section{Score operators and exact restriction of the concave lift}\label{sec:lift}

Fix $f\in\cF_{\OFL}$ and faithful states $\rho,\sigma\in\cD_d^\circ$. For $P\in\cP_d$, let $(p_j,q_j,t_j)$ be as in \eqref{eq:pqt}.

\subsection{The score operator}

\begin{definition}[Score operator]\label{def:score-operator}
The score operator generated by $P$ is
\begin{equation}\label{eq:score-operator}
 \Gamma_P=\sum_{j=1}^d s_jP_j,
 \qquad s_j=\vartheta(t_j).
\end{equation}
\end{definition}

\begin{lemma}[Exact score lift]\label{lem:exact-score-lift}
For every rank-one PVM $P$,
\begin{equation}\label{eq:exact-lift}
 F_f(P)=\cL_f^{\rho,\sigma}(\Gamma_P).
\end{equation}
Moreover, $s_j$ is the unique maximizer of the scalar concave function
\begin{equation}\label{eq:scalar-lift}
 s\longmapsto p_j\psi(s)-q_jg(s).
\end{equation}
\end{lemma}

\begin{proof}
By construction, $\psi(s_j)=f'(t_j)$. Fenchel equality \eqref{eq:fenchel-equality} gives
\begin{equation}
 q_jf(t_j)
 =p_j\psi(s_j)-q_jf^*(\psi(s_j))
 =p_j\psi(s_j)-q_jg(s_j).
\end{equation}
Summing and using functional calculus in the $P$-basis yields \eqref{eq:exact-lift}. Differentiating \eqref{eq:scalar-lift} gives
\begin{equation}\label{eq:scalar-stationarity}
 p_j\psi'(s_j)-q_jg'(s_j)=0,
\end{equation}
which is equivalent to $(f^*)'(f'(t_j))=t_j$. The scalar objective in \eqref{eq:scalar-lift} is concave by the scalar consequences of operator concavity and convexity. Strict convexity of $f$ and injectivity of $\psi$ identify the unique Fenchel optimizer.
\end{proof}

The equality \eqref{eq:exact-lift} is more than an equality of optimal values. It embeds the value of \emph{every} projective measurement into the global concave matrix problem. The embedded point is not arbitrary: its eigenvectors are the measurement basis, while its eigenvalues are the outcome-wise Fenchel optimizers.

\subsection{The matrix gradient}

Write
\begin{equation}\label{eq:G-def}
 G_P:=\nabla\cL_f^{\rho,\sigma}(\Gamma_P)
 =D\psi_{\Gamma_P}(\rho)-Dg_{\Gamma_P}(\sigma).
\end{equation}
The equality follows from self-adjointness of the Fr\'echet derivative:
\begin{equation}
 \frac{\dd}{\dd\eps}\Tr\rho\,h(\Gamma+\eps H)\Big|_{\eps=0}
 =\Tr\bigl(Dh_\Gamma(\rho)H\bigr).
\end{equation}
\begin{lemma}[Vanishing score-direction derivative]\label{lem:score-direction}
For every outcome $j$,
\begin{equation}\label{eq:diag-gradient-zero}
 \Tr(P_jG_P)=0.
\end{equation}
\end{lemma}

\begin{proof}
Since $P_j$ commutes with $\Gamma_P$, \eqref{eq:block-derivative} gives
\begin{equation}
 \Tr(P_jG_P)
 =p_j\psi'(s_j)-q_jg'(s_j),
\end{equation}
which vanishes by \eqref{eq:scalar-stationarity}.
\end{proof}

Lemma~\ref{lem:score-direction} is an envelope theorem in matrix form. It permits the score eigenvalues to depend on the measurement without contributing to the first variation. Only the rotation of the score eigenspaces survives.

\subsection{A direct commutator identity}

The same conclusion can be seen without differentiating the score map explicitly. Let
\begin{equation}
 z_j=f'(t_j)=\psi(s_j).
\end{equation}
The derivative of the perspective $qf(p/q)$ is
\begin{equation}
 \dd\bigl(qf(p/q)\bigr)
 =z_j\,\dd p-f^*(z_j)\,\dd q.
\end{equation}
For a unitary basis variation this gives
\begin{equation}\label{eq:direct-first}
 \dot F_f(P)
 =i\Tr\!\left(
 [\psi(\Gamma_P),\rho]K+[\sigma,g(\Gamma_P)]K
 \right).
\end{equation}
Using \eqref{eq:commutator-derivative},
\begin{equation}
 [\Gamma_P,D\psi_{\Gamma_P}(\rho)]=[\psi(\Gamma_P),\rho],
\end{equation}
and similarly for $g$. Thus \eqref{eq:direct-first} is identical to the gradient formula proved in the next section. This identity will be useful when checking signs and conventions.

\begin{remark}[Choice of operator coordinate]
The score operator depends on the coordinate $\psi$, whereas the projective objective does not. Different admissible coordinates reparameterize the score eigenvalues but preserve their equality pattern, because each score map $\vartheta$ is injective. The score blocks, block likelihood ratios, and the two-level escape curvature are therefore coordinate independent.
\end{remark}

%% file: sections/04-landscape.tex
\section{Critical-point geometry of the projective problem}\label{sec:landscape}

For maximization, a second-order stationary PVM is a critical PVM whose
Riemannian Hessian is negative semidefinite on the flag-manifold tangent
space.  The main theorem can now be stated exactly, before its proof is
assembled from the first-variation, rigidity, and pair-rotation results.

\begin{theorem}[Benign score-block landscape]\label{thm:benign-score-block}
Let $f\in\cF_{\OFL}$ and let $\rho,\sigma\in\cD_d^\circ$.  For a critical
rank-one PVM $P$, write
\begin{equation}\label{eq:flagship-score-blocks}
 \Gamma_P=\sum_{\alpha=1}^m s_\alpha B_\alpha,
 \qquad t_\alpha=\vartheta^{-1}(s_\alpha),
\end{equation}
with distinct scores $s_\alpha$.  Then $P$ is globally optimal over all
POVMs if and only if
\begin{equation}\label{eq:flagship-rigidity}
 B_\alpha(\rho-t_\alpha\sigma)B_\alpha=0
 \qquad\text{for every }\alpha.
\end{equation}
If this condition fails, some $B_\alpha$ contains basis vectors $j\ne k$
for which a phase-adjusted two-level rotation satisfies
\begin{align}
 \frac{\dd^2}{\dd\theta^2}F_f(P(\theta))\Big|_{\theta=0}
 &=4f''(t_\alpha)
 \left|\bra j(\rho-t_\alpha\sigma)\ket k\right|^2\notag\\
 &\quad\times\left(\frac1{q_j}+\frac1{q_k}\right)>0.
 \label{eq:flagship-curvature}
\end{align}
Consequently every projective local maximum and every second-order
stationary PVM is globally optimal over all POVMs, and every nonglobal
critical PVM has an explicit positive-curvature pair rotation.
\end{theorem}

\subsection{First variation}

\begin{theorem}[Commutator form of the first variation]\label{thm:first-variation}
Let $P(\eps)=e^{i\eps K}Pe^{-i\eps K}$ for a Hermitian generator $K$. Then
\begin{equation}\label{eq:first-variation}
 \frac{\dd}{\dd\eps}F_f(P(\eps))\Big|_{\eps=0}
 =i\Tr\!\left([\Gamma_P,G_P]K\right).
\end{equation}
Consequently $P$ is a critical point of $F_f$ on $\cP_d$ if and only if
\begin{equation}\label{eq:stationarity-commutator}
 [\Gamma_P,G_P]=0.
\end{equation}
\end{theorem}

\begin{proof}
Along the path, write
\begin{equation}
 \Gamma_{P(\eps)}=\sum_js_j(\eps)P_j(\eps).
\end{equation}
At $\eps=0$,
\begin{equation}
 \dot\Gamma=i[K,\Gamma_P]+\sum_j\dot s_jP_j.
\end{equation}
By Lemma~\ref{lem:exact-score-lift},
\begin{equation}
 \frac{\dd}{\dd\eps}F_f(P(\eps))\Big|_0=\Tr(G_P\dot\Gamma).
\end{equation}
The second term vanishes by Lemma~\ref{lem:score-direction}. Cyclicity of trace gives
\begin{equation}
 \Tr(G_P i[K,\Gamma_P])=i\Tr([\Gamma_P,G_P]K),
\end{equation}
which proves \eqref{eq:first-variation}. Since $i[\Gamma_P,G_P]$ is Hermitian and $K$ is arbitrary, the derivative vanishes for every tangent direction exactly when \eqref{eq:stationarity-commutator} holds.
\end{proof}

\begin{remark}
Equation~\eqref{eq:stationarity-commutator} is stronger than a list of scalar stationarity equations: it identifies the obstruction as a failure of the score operator to commute with the gradient of the global concave problem. When the score spectrum has degeneracies, rotations inside an equal-score block are invisible at first order. The second-order analysis below resolves precisely those directions.
\end{remark}

\subsection{Score blocks}

Let
\begin{equation}\label{eq:score-blocks}
 \Gamma_P=\sum_{\alpha=1}^m s_\alpha B_\alpha
\end{equation}
be the spectral decomposition with distinct $s_\alpha$, and define
\begin{equation}\label{eq:block-likelihood}
 t_\alpha=\vartheta^{-1}(s_\alpha).
\end{equation}
Thus $B_\alpha$ is the sum of those rank-one projectors whose classical likelihood ratio equals $t_\alpha$.

\begin{lemma}[Block gradient identity]\label{lem:block-gradient}
For every score block,
\begin{equation}\label{eq:block-gradient}
 B_\alpha G_PB_\alpha
 =\psi'(s_\alpha)B_\alpha(\rho-t_\alpha\sigma)B_\alpha.
\end{equation}
\end{lemma}

\begin{proof}
Equation~\eqref{eq:block-derivative} gives
\begin{equation}
 B_\alpha G_PB_\alpha
 =\psi'(s_\alpha)B_\alpha\rho B_\alpha
 -g'(s_\alpha)B_\alpha\sigma B_\alpha.
\end{equation}
The scalar optimality equation at $s_\alpha$ is
\begin{equation}
 t_\alpha\psi'(s_\alpha)=g'(s_\alpha),
\end{equation}
so \eqref{eq:block-gradient} follows.
\end{proof}

We call
\begin{equation}\label{eq:block-residual}
 A_\alpha(P)=B_\alpha(\rho-t_\alpha\sigma)B_\alpha
\end{equation}
the \emph{score-block residual}. Its diagonal entries in the $P$-basis vanish:
\begin{equation}\label{eq:zero-diagonal}
 \bra jA_\alpha(P)\ket j=p_j-t_\alpha q_j=0,
 \qquad P_j\le B_\alpha.
\end{equation}
Thus a nonzero residual is necessarily noncommutative: it is carried by off-diagonal matrix elements inside a degenerate score block.

\subsection{Rigidity and global optimality}

\begin{theorem}[Score-block rigidity]\label{thm:block-rigidity}
Let $f\in\cF_{\OFL}$ and $\rho,\sigma\in\cD_d^\circ$. For a rank-one PVM $P$, the following are equivalent:
\begin{enumerate}
\item $P$ is globally optimal for $D_f^{\mathbb M}(\rho\Vert\sigma)$;
\item $P$ is critical and
\begin{equation}\label{eq:rigidity-condition}
 A_\alpha(P)=0\quad\text{for every score block }\alpha;
\end{equation}
\item $G_P=0$.
\end{enumerate}
\end{theorem}

\begin{proof}
Assume (2). Criticality gives $[\Gamma_P,G_P]=0$. Therefore, for $\alpha\ne\beta$,
\begin{equation}
 (s_\alpha-s_\beta)B_\alpha G_PB_\beta=0,
\end{equation}
so all off-diagonal score blocks of $G_P$ vanish. Lemma~\ref{lem:block-gradient} and \eqref{eq:rigidity-condition} show that every diagonal score block vanishes. Hence $G_P=0$, proving (3).

If (3) holds, concavity of $\cL_f^{\rho,\sigma}$ on $\Omega_J$ yields
\begin{equation}
 \cL_f^{\rho,\sigma}(\Gamma)
 \le \cL_f^{\rho,\sigma}(\Gamma_P)
 +\Tr(G_P(\Gamma-\Gamma_P))
 =\cL_f^{\rho,\sigma}(\Gamma_P)
\end{equation}
for all $\Gamma\in\Omega_J$. Lemma~\ref{lem:exact-score-lift} and Theorem~\ref{thm:fff} show that $P$ attains the full POVM optimum, proving (1).

Finally assume (1). Then
\begin{equation}
 \cL_f^{\rho,\sigma}(\Gamma_P)=F_f(P)
 =\sup_{\Gamma\in\Omega_J}\cL_f^{\rho,\sigma}(\Gamma).
\end{equation}
Because $\Gamma_P$ lies in the open spectral domain and the functional is differentiable, $G_P=0$. This proves (3), and Lemma~\ref{lem:block-gradient}, together with $\psi'(s_\alpha)\ne0$, gives \eqref{eq:rigidity-condition}. Equation~\eqref{eq:first-variation} also gives criticality, so (2) follows.
\end{proof}

\begin{corollary}[Simple-score critical points]\label{cor:simple-score}
If $P$ is critical and all likelihood ratios $t_j$ are distinct, then $P$ is globally optimal over all POVMs.
\end{corollary}

\begin{proof}
Every score block is rank one, and \eqref{eq:zero-diagonal} makes each block residual zero. Apply Theorem~\ref{thm:block-rigidity}.
\end{proof}

\subsection{The two-level escape direction}

For $p,q>0$, define the perspective
\begin{equation}\label{eq:perspective}
 \varphi_f(p,q)=qf(p/q).
\end{equation}
A direct calculation gives
\begin{equation}\label{eq:perspective-hessian}
 \dd^2\varphi_f
 =\frac{f''(t)}q(\dd p-t\,\dd q)^2,
 \qquad t=p/q.
\end{equation}
The Hessian has rank one; its null direction is the radial rescaling $(p,q)\mapsto c(p,q)$.

\begin{lemma}[Pair-rotation second variation]\label{lem:pair-rotation}
Suppose two outcomes $j,k$ have the same likelihood ratio $t$. Let $P(\theta)$ be the two-level rotation
\begin{align}
 \ket{j(\theta)}&=\cos\theta\ket j+e^{i\phi}\sin\theta\ket k,\label{eq:two-level-j}\\
 \ket{k(\theta)}&=-e^{-i\phi}\sin\theta\ket j+\cos\theta\ket k,
 \label{eq:two-level-k}
\end{align}
with all other vectors fixed. Then
\begin{equation}\label{eq:pair-second-general}
 \frac{\dd^2}{\dd\theta^2}F_f(P(\theta))\Big|_0
 =\frac{f''(t)}{q_j}(\dot p_j-t\dot q_j)^2
 +\frac{f''(t)}{q_k}(\dot p_k-t\dot q_k)^2.
\end{equation}
\end{lemma}

\begin{proof}
Only the $j$ and $k$ perspective terms vary. The rotation preserves $p_j+p_k$ and $q_j+q_k$. Since both points have the same ratio $t$, their perspective gradients coincide. The gradient terms multiplying the second accelerations therefore cancel pairwise. The remaining quadratic terms are exactly \eqref{eq:perspective-hessian} applied to the two outcomes.
\end{proof}

\begin{theorem}[Explicit strict-ascent direction]\label{thm:strict-ascent}
Let $P$ be a critical point which is not globally optimal. Then some score block $B_\alpha$ contains indices $j\ne k$ such that
\begin{equation}
 a_{jk}:=\bra j(\rho-t_\alpha\sigma)\ket k\ne0.
\end{equation}
There is a phase $\phi$ for which the rotation \eqref{eq:two-level-j}--\eqref{eq:two-level-k} satisfies
\begin{equation}\label{eq:strict-curvature}
 \frac{\dd^2}{\dd\theta^2}F_f(P(\theta))\Big|_{\theta=0}
 =4f''(t_\alpha)|a_{jk}|^2
 \left(\frac1{q_j}+\frac1{q_k}\right)>0.
\end{equation}
\end{theorem}

\begin{proof}
By Theorem~\ref{thm:block-rigidity}, some block residual $A_\alpha(P)$ is nonzero. It is Hermitian and has zero diagonal by \eqref{eq:zero-diagonal}; hence it has a nonzero off-diagonal entry $a_{jk}$.

Choose $\phi$ so that
\begin{equation}
 \Re\!\left(e^{i\phi}a_{jk}\right)=|a_{jk}|.
\end{equation}
The chosen rotation yields
\begin{equation}
 \dot p_j-t_\alpha\dot q_j=2|a_{jk}|,
 \qquad
 \dot p_k-t_\alpha\dot q_k=-2|a_{jk}|.
\end{equation}
Substitution into Lemma~\ref{lem:pair-rotation} proves \eqref{eq:strict-curvature}.
\end{proof}

\begin{proof}[Proof of Theorem~\ref{thm:benign-score-block}]
Theorem~\ref{thm:block-rigidity} gives the equivalence between global
optimality and \eqref{eq:flagship-rigidity}.  If the block condition fails,
Theorem~\ref{thm:strict-ascent} supplies the rotation and
\eqref{eq:flagship-curvature}.  A local maximum is critical and has
nonpositive second variation in every direction, so the positive-curvature
alternative is impossible.  The same argument applies to every
second-order stationary PVM.
\end{proof}

\subsection{Finite exact optimality test}

At an exact critical PVM, the proof is a finite procedure.

\begin{algorithm}[Score-block optimality test]\label{alg:score-block}
\emph{Input:} a critical rank-one PVM $P$ for faithful $\rho,\sigma$ and a
smoothly operator-Fenchel liftable $f$.
\begin{enumerate}
\item Compute $t_j=p_j/q_j$, the scores $s_j=\vartheta(t_j)$, and the
distinct-score blocks $B_\alpha$.
\item Form $A_\alpha=B_\alpha(\rho-t_\alpha\sigma)B_\alpha$.
\item If every $A_\alpha$ vanishes, return \emph{globally POVM-optimal}.
\item Otherwise choose an off-diagonal entry
$a_{jk}=\bra jA_\alpha\ket k\ne0$, take
$\phi=-\arg a_{jk}$, and return the pair rotation
\eqref{eq:two-level-j}--\eqref{eq:two-level-k} with ascent curvature
\eqref{eq:strict-curvature}.
\end{enumerate}
\emph{Correctness:} Theorems~\ref{thm:block-rigidity} and
\ref{thm:strict-ascent}.
\end{algorithm}

\subsection{A three-dimensional score-block saddle}

The escape mechanism is already visible in dimension three.

\begin{example}[A genuinely degenerate score block]\label{ex:qutrit-saddle}
For measured relative entropy, $f(t)=t\log t$, let
\begin{equation}\label{eq:qutrit-states}
 \sigma=\begin{pmatrix}
  1/4&0&0\\0&1/4&0\\0&0&1/2
 \end{pmatrix},\qquad
 \rho=\begin{pmatrix}
  3/10&1/10&0\\1/10&3/10&0\\0&0&2/5
 \end{pmatrix}.
\end{equation}
Both states are faithful.  At the standard PVM,
\begin{equation}\label{eq:qutrit-ratios}
 (t_1,t_2,t_3)=(6/5,6/5,4/5).
\end{equation}
Using $\psi(s)=1+\log s$, so that $\vartheta(t)=t$, gives
\begin{equation}\label{eq:qutrit-score-gradient}
 \Gamma_P=\operatorname{diag}(6/5,6/5,4/5),\qquad
 G_P=\begin{pmatrix}
  0&1/12&0\\1/12&0&0\\0&0&0
 \end{pmatrix}.
\end{equation}
Thus $[\Gamma_P,G_P]=0$, so $P$ is critical, while the degenerate block
$B=P_1+P_2$ has residual
\begin{equation}\label{eq:qutrit-residual}
 B(\rho-\tfrac65\sigma)B
 =\begin{pmatrix}0&1/10\\1/10&0\end{pmatrix}\ne0.
\end{equation}
The $1$--$2$ pair rotation therefore has
\begin{equation}\label{eq:qutrit-positive-curvature}
 F_{\theta\theta}(0,0)
 =4\frac56\left(\frac1{10}\right)^2(4+4)=\frac4{15}>0.
\end{equation}
For reproducibility, put
\begin{equation}
 \begin{aligned}
 u_\theta&=\cos\theta\,e_1+\sin\theta\,e_2,
 &v_\theta&=-\sin\theta\,e_1+\cos\theta\,e_2,\\
 u_{\theta,\eta}&=\cos\eta\,u_\theta+\sin\eta\,e_3,
 &w_{\theta,\eta}&=-\sin\eta\,u_\theta+\cos\eta\,e_3,
 \end{aligned}
\end{equation}
and let $P(\theta,\eta)$ project onto
$(u_{\theta,\eta},v_\theta,w_{\theta,\eta})$.  Direct evaluation gives
\begin{equation}\label{eq:qutrit-saddle-expansion}
 \begin{aligned}
 F(\theta,\eta)-F(0,0)
 &=\frac{2}{15}\theta^2\\
 &\quad+\frac1{10}\left(\log\frac32-1\right)\eta^2
 +O(\|(\theta,\eta)\|^3).
 \end{aligned}
\end{equation}
The second coefficient is negative.  Hence the standard PVM is a literal
sign-changing saddle: the hidden noncommutativity in the equal-score block
selects the ascent direction, while mixing distinct scores is descending.
\end{example}

\begin{figure}[!t]
 \centering
 \includegraphics[width=0.78\textwidth]{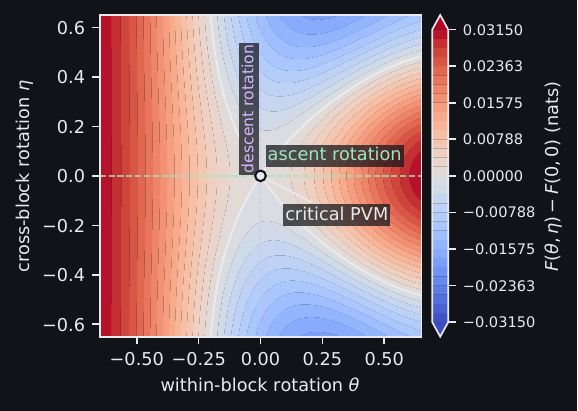}
 \caption{Relative-entropy landscape in Example~\ref{ex:qutrit-saddle}.
 The horizontal pair rotation stays inside the degenerate score block and
 increases the objective; the vertical cross-block rotation decreases it.
 The white contour is the critical objective value.}
 \label{fig:qutrit-saddle}
\end{figure}

The general conclusion is the \emph{strict-ascent property}: a nonglobal
critical point has a positive Hessian direction in the maximization
convention.  It need not have a negative direction as well; thus
``strict saddle'' below always refers to this one-sided optimization
convention.  Example~\ref{ex:qutrit-saddle} happens to be sign-changing.

\subsection{A quantitative ascent measure}

The proof gives a computable defect at every critical point. Define
\begin{equation}\label{eq:escape-index}
 \mathfrak e_f(P)=
 \max_{\substack{\alpha,\,j\ne k\\P_j,P_k\le B_\alpha}}
 4f''(t_\alpha)
 \left|\bra j(\rho-t_\alpha\sigma)\ket k\right|^2
 \left(\frac1{q_j}+\frac1{q_k}\right).
\end{equation}
Then $\mathfrak e_f(P)=0$ if and only if the critical point is global, and otherwise it is the largest second variation attained by the unit-angle pair rotations used in Theorem~\ref{thm:strict-ascent}. In particular it identifies a positive Hessian direction, without choosing a global normalization convention for the flag-manifold metric. This defect separates three logically different notions:
\begin{enumerate}
\item the global value gap, which need not be known;
\item the first-order residual $[\Gamma_P,G_P]$, which vanishes at every critical point;
\item the hidden block residuals, which decide whether a critical point is globally optimal.
\end{enumerate}

\begin{remark}[Scope of the theorem]
Theorem~\ref{thm:benign-score-block} classifies critical PVMs and identifies
an exact ascent direction at every nonglobal one.  Convergence guarantees
for first-order methods and dimension-independent saddle gaps require
additional quantitative control of the smallest nonzero block residual and
of the scalar or atomic curvature, as discussed in
Section~\ref{sec:outlook}.
\end{remark}

%% file: sections/05-atomic.tex
\section{Positive atomic curvature}\label{sec:atomic}

\subsection{The resolvent atoms}

For $\lambda\in[0,1]$, define
\begin{equation}\label{eq:atom-generator}
 h_\lambda(t)=\frac{(t-1)^2}{(1-\lambda)t+\lambda},
 \qquad t>0.
\end{equation}
These are the scalar generators of the resolvent $\chi^2_\lambda$ divergences. Their curvature is
\begin{equation}\label{eq:atom-second}
 h_\lambda''(t)=\frac{2}{((1-\lambda)t+\lambda)^3}.
\end{equation}

The operator-convex representation is most naturally written with a positive measure; this allows endpoint atoms and avoids an unnecessary absolute-continuity assumption.

\begin{definition}[Positively atomizable generator]\label{def:atomizable}
A normalized convex generator $f$ is \emph{positively atomizable} if there is a finite positive Borel measure $\nu_f$ on $[0,1]$ such that, up to an affine term that does not affect probability divergences,
\begin{equation}\label{eq:atomic-representation}
 f(t)=\int_{[0,1]}h_\lambda(t)\,\nu_f(\dd\lambda).
\end{equation}
\end{definition}

The standard L\"owner representation of operator-convex functions on
$(0,\infty)$ implies such a measure form after affine normalization
\cite{HansenPedersen1982,BhatiaMatrix1997}.  Salazar formulated the
corresponding positive resolvent-atom decomposition for Petz quantum
$f$-divergences \cite{SalazarAtomic2026}.  We use its scalar consequence
\eqref{eq:atomic-representation}.

\begin{lemma}[Curvature transform]\label{lem:curvature-transform}
If $f$ is positively atomizable, then
\begin{equation}\label{eq:f-second-atomic}
 f''(t)=2\int_{[0,1]}
 \frac{\nu_f(\dd\lambda)}{((1-\lambda)t+\lambda)^3}.
\end{equation}
In particular, a nonzero positive atomic measure makes $f$ strictly convex.
\end{lemma}

\begin{proof}
Differentiate \eqref{eq:atomic-representation} twice and use \eqref{eq:atom-second}. Differentiation under the integral is justified on compact subintervals of $(0,\infty)$ by the finiteness of $\nu_f$ and the uniform lower bound on the denominator.
\end{proof}

\subsection{Atomic resolution of the escape curvature}

\begin{theorem}[Atomic strict-ascent curvature]\label{thm:atomic-curvature}
Assume the hypotheses of Theorem~\ref{thm:strict-ascent} and suppose in addition that $f$ is positively atomizable. For the two-level escape direction found there,
\begin{align}
 \frac{\dd^2}{\dd\theta^2}F_f(P(\theta))\Big|_{\theta=0}
 &={}8|a_{jk}|^2
 \left(\frac1{q_j}+\frac1{q_k}\right)\notag\\
 &\quad\times\int_{[0,1]}
 \frac{\nu_f(\dd\lambda)}{((1-\lambda)t_\alpha+\lambda)^3}.
 \label{eq:atomic-curvature}
\end{align}
Equivalently, the curvature of $F_f$ is the positive superposition of the curvatures obtained by replacing $f$ with each atom $h_\lambda$.
\end{theorem}

\begin{proof}
Substitute \eqref{eq:f-second-atomic} into \eqref{eq:strict-curvature}.
\end{proof}

\begin{corollary}[Common escape rotation for positive atomic families]\label{cor:universal-escape}
Let $\mathcal C$ be any collection of positively atomizable generators whose atomic measures are nonzero. At a projective measurement $P$ which is critical for every $f\in\mathcal C$, if one common score block has a nonzero residual $A_\alpha(P)$, then the same pair rotation increases every objective in $\mathcal C$ to second order.
\end{corollary}

\begin{proof}
The residual chooses $j,k$ and the phase independently of the atomic measure. Every measure contributes positively in \eqref{eq:atomic-curvature}.
\end{proof}

The score partition itself is generator-independent: injectivity of every
score map makes its blocks exactly the equal-likelihood-ratio blocks.
Corollary~\ref{cor:universal-escape} additionally requires the same PVM to
satisfy the first-order equations for every objective.  Once such a common
critical PVM is fixed, the noncommutative residual determines the escape
direction, and the divergence supplies a positive scalar weighting of its
curvature.

\begin{remark}[Why atomwise maximization is insufficient]
For each $\lambda$, the atomic projective problem can have an optimizer whose basis depends on $\lambda$. In general,
\begin{equation}
 \sup_P\int F_{h_\lambda}(P)\,\nu_f(\dd\lambda)
 \le
 \int\sup_PF_{h_\lambda}(P)\,\nu_f(\dd\lambda),
\end{equation}
and equality need not hold. The proof of the landscape theorem does not exchange supremum and integration. It first resums the scalar Fenchel information into $\Gamma_P$, then uses the atomic representation only to resolve the already identified Hessian direction.
\end{remark}

\subsection{A global quadratic error bound}

The endpoint atom $h_1(t)=(t-1)^2$, extended to all of $\RR$ by the same
polynomial, admits a stronger quantitative statement.  Choose the affine
lift coordinate
\begin{equation}\label{eq:quadratic-coordinate}
 \psi(s)=2(s-1),\qquad g(s)=s^2-1.
\end{equation}
Then $\vartheta(t)=t$ and
\begin{equation}\label{eq:quadratic-lift}
 \cL_{h_1}^{\rho,\sigma}(\Gamma)
 =2\Tr(\rho\Gamma)-\Tr(\sigma\Gamma^2)-1.
\end{equation}

\begin{theorem}[Global $\chi^2$ residual equivalence]
\label{thm:quadratic-robust}
Let $\rho,\sigma\in\cD_d^\circ$ and $P\in\cP_d$.  Put
\begin{equation}\label{eq:quadratic-score-gradient}
 \Gamma_P=\sum_jt_jP_j,\qquad
 G_P=2\rho-\sigma\Gamma_P-\Gamma_P\sigma,
\end{equation}
and let $\Gamma_\star$ be the unique Hermitian solution of
\begin{equation}\label{eq:quadratic-global-score}
 \sigma\Gamma_\star+\Gamma_\star\sigma=2\rho.
\end{equation}
For the distinct-score decomposition
$\Gamma_P=\sum_{\alpha=1}^m t_\alpha B_\alpha$, define
\begin{align}
 r_P^2&=\sum_{\alpha=1}^m
 \left\|B_\alpha(\rho-t_\alpha\sigma)B_\alpha\right\|_2^2,
 \label{eq:quadratic-block-residual}\\
 c_P&=\|[\Gamma_P,G_P]\|_2.
 \label{eq:quadratic-commutator}
\end{align}
Here $\|\cdot\|_2$ is the Hilbert--Schmidt norm.  When $m\ge2$, define
\begin{equation}\label{eq:quadratic-score-gaps}
 \delta_P=\min_{\alpha\ne\beta}|t_\alpha-t_\beta|,\qquad
 \Delta_P=\max_{\alpha\ne\beta}|t_\alpha-t_\beta|.
\end{equation}
If $m=1$, then $c_P=0$, and every quotient involving $\delta_P$ or
$\Delta_P$ below is defined to be zero.
Then
\begin{align}
 \frac{4r_P^2+c_P^2/\Delta_P^2}
      {4\lambda_{\max}(\sigma)}
 &\le D_{h_1}^{\mathbb M}(\rho\Vert\sigma)-F_{h_1}(P)
 \notag\\
 &\le
 \frac{4r_P^2+c_P^2/\delta_P^2}
      {4\lambda_{\min}(\sigma)},
 \label{eq:quadratic-value-two-sided}
\end{align}
and
\begin{align}
 \frac{\sqrt{4r_P^2+c_P^2/\Delta_P^2}}
      {2\lambda_{\max}(\sigma)}
 &\le \|\Gamma_P-\Gamma_\star\|_2
 \notag\\
 &\le
 \frac{\sqrt{4r_P^2+c_P^2/\delta_P^2}}
      {2\lambda_{\min}(\sigma)}.
 \label{eq:quadratic-score-two-sided}
\end{align}
Consequently, $r_P=c_P=0$ if and only if $P$ is globally optimal.
\end{theorem}

\begin{proof}
The gradient of \eqref{eq:quadratic-lift} is
$G(\Gamma)=2\rho-\sigma\Gamma-\Gamma\sigma$, so
\eqref{eq:quadratic-global-score} is its unique stationarity equation.
Writing $H=\Gamma_P-\Gamma_\star$ and completing the square gives
\begin{equation}\label{eq:quadratic-exact-gap}
 D_{h_1}^{\mathbb M}(\rho\Vert\sigma)-F_{h_1}(P)
 =\Tr(\sigma H^2),
 \qquad G_P=-(\sigma H+H\sigma).
\end{equation}
In an eigenbasis of $\sigma$,
\begin{align}
 4\lambda_{\min}(\sigma)\Tr(\sigma H^2)
 &\le \|G_P\|_2^2
 \le4\lambda_{\max}(\sigma)\Tr(\sigma H^2),
 \label{eq:quadratic-gap-gradient}\\
 2\lambda_{\min}(\sigma)\|H\|_2
 &\le\|G_P\|_2
 \le2\lambda_{\max}(\sigma)\|H\|_2.
 \label{eq:quadratic-score-gradient-bound}
\end{align}
The diagonal score blocks satisfy
\begin{equation}\label{eq:quadratic-diagonal-block}
 B_\alpha G_PB_\alpha
 =2B_\alpha(\rho-t_\alpha\sigma)B_\alpha,
\end{equation}
whereas, for $\alpha\ne\beta$,
\begin{equation}\label{eq:quadratic-off-block}
 B_\alpha[\Gamma_P,G_P]B_\beta
 =(t_\alpha-t_\beta)B_\alpha G_PB_\beta.
\end{equation}
Orthogonality of the matrix blocks therefore gives
\begin{equation}\label{eq:quadratic-residual-gradient}
 4r_P^2+\frac{c_P^2}{\Delta_P^2}
 \le\|G_P\|_2^2
 \le4r_P^2+\frac{c_P^2}{\delta_P^2}.
\end{equation}
Combining \eqref{eq:quadratic-exact-gap}--
\eqref{eq:quadratic-residual-gradient} proves both bounds.
\end{proof}

\begin{corollary}[Separated-score flag bound]
\label{cor:quadratic-flag-bound}
Suppose the scores of $P$ are simple and
$c_P<\lambda_{\min}(\sigma)\delta_P^2$.  Then there is an optimal
rank-one PVM $P^\star=(P_j^\star)_j$, after relabeling, such that
\begin{equation}\label{eq:quadratic-flag-distance}
 d_{\rm ch}(P,P^\star)
 :=\left(\frac12\sum_j\|P_j-P_j^\star\|_2^2\right)^{1/2}
 \le\frac{\sqrt d\,c_P}
 {\lambda_{\min}(\sigma)\delta_P^2}.
\end{equation}
\end{corollary}

\begin{proof}
Simple scores imply $r_P=0$.  Equations
\eqref{eq:quadratic-score-two-sided} and
\eqref{eq:quadratic-score-gaps} give
$\|\Gamma_P-\Gamma_\star\|_\infty<\delta_P/2$ and
\begin{equation}\label{eq:quadratic-score-operator-bound}
 \|\Gamma_P-\Gamma_\star\|_\infty
 \le\frac{c_P}{2\lambda_{\min}(\sigma)\delta_P}.
\end{equation}
The standard spectral-subspace perturbation bound
\cite[Ch.~VII]{BhatiaMatrix1997} matches the rank-one spectral projectors
of $\Gamma_P$ and $\Gamma_\star$ and gives
$d_{\rm ch}\le2\sqrt d\,
\|\Gamma_P-\Gamma_\star\|_\infty/\delta_P$.
Weyl's inequality also shows that $\Gamma_\star$ has simple spectrum.
Let $P_j^\star$ be its matched spectral projectors, with eigenvalues
$\gamma_j$.  Taking diagonal matrix elements of
\eqref{eq:quadratic-global-score} gives
$\Tr(\rho P_j^\star)=\gamma_j\Tr(\sigma P_j^\star)$, so
$\Gamma_{P^\star}=\Gamma_\star$.  The exact lift and
Theorem~\ref{thm:fff} therefore make $P^\star$ globally optimal.
Substitution proves
\eqref{eq:quadratic-flag-distance}.
\end{proof}

%% file: sections/09-renyi.tex
\section{Measured R\'enyi and relative-entropy landscapes}\label{sec:renyi}

\subsection{Power objectives}

For classical probability distributions with common finite support, define
\begin{equation}\label{eq:renyi-power}
 Q_\alpha(p\Vert q)=\sum_y p_y^\alpha q_y^{1-\alpha},
 \qquad \alpha>0,
 \qquad \alpha\ne1,
\end{equation}
and
\begin{equation}\label{eq:classical-renyi}
 D_\alpha(p\Vert q)=\frac1{\alpha-1}\log Q_\alpha(p\Vert q).
\end{equation}
The measured R\'enyi divergence is the supremum of \eqref{eq:classical-renyi} over measurements. At $\alpha=1$ it converges to the measured relative entropy.

The relevant strictly convex generators are
\begin{equation}\label{eq:renyi-generators}
 f_\alpha(t)=
 \begin{cases}
 1-t^\alpha,&0<\alpha<1,\\
 t\log t,&\alpha=1,\\
 t^\alpha-1,&\alpha>1.
 \end{cases}
\end{equation}
For Definition~\ref{def:ofl}, these formulas are supplied with their
standard global lower-semicontinuous convex extensions: $f_\alpha=+\infty$
on $(-\infty,0)$ when $0<\alpha\le1$, with the usual finite value at zero,
and $f_\alpha(t)=(\max\{t,0\})^\alpha-1$ when $\alpha>1$.  This choice is
part of the data because the effective domain of the Fenchel conjugate
depends on it.
Indeed,
\begin{equation}\label{eq:renyi-second}
 f_\alpha''(t)=
 \begin{cases}
 \alpha(1-\alpha)t^{\alpha-2},&0<\alpha<1,\\
 t^{-1},&\alpha=1,\\
 \alpha(\alpha-1)t^{\alpha-2},&\alpha>1.
 \end{cases}
\end{equation}
For $0<\alpha<1$ one has $F_{f_\alpha}=1-Q_\alpha$, whereas for $\alpha>1$ one has $F_{f_\alpha}=Q_\alpha-1$. Consequently
\begin{equation}
 D_\alpha=
 \begin{cases}
 \dfrac1{\alpha-1}\log(1-F_{f_\alpha}),&0<\alpha<1,\\[2mm]
 \dfrac1{\alpha-1}\log(1+F_{f_\alpha}),&\alpha>1,
 \end{cases}
\end{equation}
and in both regimes the displayed scalar function is strictly increasing on its natural domain.

Fang, Fawzi, and Fawzi verify the operator-Fenchel condition for every
finite order $\alpha>0$ and for relative entropy
\cite{FangFawziFawzi2026}; they also note that the corresponding measured
R\'enyi and relative-entropy variational formulas were already established
by Berta, Fawzi, and Tomamichel
\cite{BertaFawziTomamichel2017}.  Thus each generator in
\eqref{eq:renyi-generators} belongs to the imported liftable class used
here.

\subsection{The landscape theorem}

\begin{theorem}[Measured R\'enyi landscape]\label{thm:renyi-landscape}
Let $\rho,\sigma\in\cD_d^\circ$ and let $\alpha>0$, $\alpha\ne1$, be finite.
For the projective optimization of the measured R\'enyi divergence of order
$\alpha$:
\begin{enumerate}
\item every local maximum over rank-one PVMs is globally optimal over all POVMs;
\item every second-order stationary PVM is globally optimal over all POVMs;
\item every nonglobal critical PVM admits an explicit two-level direction of strictly increasing measured R\'enyi divergence.
\end{enumerate}
The same statements hold for measured relative entropy at $\alpha=1$.
\end{theorem}

\begin{proof}
Apply Theorem~\ref{thm:benign-score-block} to the appropriate generator in \eqref{eq:renyi-generators}. The R\'enyi divergence is a smooth strictly increasing scalar function of $F_{f_\alpha}$. Hence the two objectives have the same critical points and local maxima. At a critical point, the Hessian of the transformed objective is the positive derivative of the scalar transformation times the Hessian of $F_{f_\alpha}$, so the strict-ascent direction is preserved.
\end{proof}

Theorem~\ref{thm:renyi-landscape} goes beyond the imported equality of
optimal projective and POVM values by adding a local-to-global
classification of projective critical points and an explicit
strict-ascent direction.

\subsection{Atomic regimes}

The power generators in \eqref{eq:renyi-generators} are operator convex precisely in the ranges relevant here as follows:
\begin{equation}\label{eq:operator-convex-power-ranges}
 \begin{aligned}
 1-t^\alpha&\ \text{ is operator convex},&&0<\alpha<1,\\
 t^\alpha-1&\ \text{ is operator convex},&&1<\alpha\le2.
 \end{aligned}
\end{equation}
At $\alpha=1$, the normalized family satisfies
\begin{equation}\label{eq:power-relative-limit}
 \frac{t^\alpha-t}{\alpha-1}\longrightarrow t\log t,
\end{equation}
where replacing $1$ by $t$ changes the generator only by an affine term and
hence leaves its probability divergence unchanged.  The limit is operator
convex.  Therefore the strict-ascent curvature has a
positive atomic resolution for every $0<\alpha\le2$.

For $\alpha>2$, the operator-Fenchel lift still gives the benign projective
landscape.  The power generator is outside the operator-convex range, so
the positive $\chi^2_\lambda$ representation used in
Section~\ref{sec:atomic} is unavailable.  This separates two logically
independent structures:
\begin{enumerate}
\item \emph{liftability}, which controls the global geometry of critical points;
\item \emph{positive atomizability}, which resolves the scalar curvature into universal resolvent atoms.
\end{enumerate}
Neither property implies the other in the generality needed for a complete classification.

\subsection{Likelihood-score form of the critical equations}

Although a coordinate-free theorem is preferable, the block condition has a particularly simple interpretation for power divergences. At a critical PVM, group outcomes with the same likelihood ratio $t_\alpha$. Global optimality is equivalent to
\begin{equation}\label{eq:renyi-block}
 B_\alpha\rho B_\alpha=t_\alpha B_\alpha\sigma B_\alpha
 \qquad\text{for every }\alpha.
\end{equation}
Thus every equal-likelihood subspace must be a subspace on which the compressed states are exactly proportional. If the relation fails, any nonzero off-diagonal entry of the defect in \eqref{eq:renyi-block} gives the explicit pair rotation of Theorem~\ref{thm:strict-ascent}.

Because every score map is injective, the score partition is itself the equal-likelihood-ratio partition and is independent of $\alpha$. What changes with the order is the first-order commutator equation and the scalar curvature factor in \eqref{eq:renyi-second}. This suggests comparing R\'enyi orders through common critical PVMs, their shared likelihood blocks, and their optimal values.

%% file: sections/03-shor.tex
\section{The binary orthogonal-measurement assertion as a projectivization corollary}\label{sec:shor}

Let
\begin{equation}\label{eq:binary-ensemble}
 \mathcal E=\{(a,\rho_0),(b,\rho_1)\},
 \qquad a,b>0,\qquad a+b=1,
\end{equation}
be a binary quantum ensemble.  A measurement $M=(M_y)_y$ gives
$p_y=\Tr(\rho_0M_y)$, $q_y=\Tr(\rho_1M_y)$, and
$m_y=ap_y+bq_y$.  Its mutual information is
\begin{equation}\label{eq:mutual-measurement}
 I_{\mathcal E}(M)=\sum_y\left[
 ap_y\log\frac{p_y}{m_y}+bq_y\log\frac{q_y}{m_y}\right].
\end{equation}

The connection with the imported projectivization theorem is immediate.
The scalar objective is the weighted Jensen--Shannon divergence
\cite{Lin1991}.
Wang, Wang, and Chen also recorded the same implication in their
August 2026 posterior-algebra treatment \cite{WangWangChen2026}; the point
of isolating it here is to separate the prior existence theorem from the
stronger landscape statement proved below.

\begin{proposition}[Weighted Jensen--Shannon lift]\label{prop:js-lift}
Define, on $(0,\infty)$,
\begin{equation}\label{eq:js-generator}
 f_{a,b}(t)=at\log t-(at+b)\log(at+b),\qquad t>0.
\end{equation}
Take its lower-semicontinuous convex extension with
$f_{a,b}(0)=-b\log b$ and $f_{a,b}(t)=+\infty$ for $t<0$.  Then
$I_{\mathcal E}(M)=D_{f_{a,b}}(p\Vert q)$ and
\begin{equation}\label{eq:js-second}
 f_{a,b}''(t)=\frac{ab}{t(at+b)}>0.
\end{equation}
Moreover,
\begin{align}
 f_{a,b}^*(z)&=b\log\frac{b}{1-ae^{z/a}},
   &&z<-a\log a,\label{eq:js-conjugate}\\
 \psi(s)&=a\log s,
   &&0<s<1/a,\label{eq:js-psi}\\
 (f_{a,b}^*\!\circ\psi)(s)&=b\log\frac{b}{1-as}.
   \label{eq:js-g}
\end{align}
Thus $\psi$ is operator concave, $f_{a,b}^*\!\circ\psi$ is operator
convex, and $f_{a,b}\in\cF_{\OFL}$ with score map
\begin{equation}\label{eq:js-score}
 \vartheta(t)=\frac{t}{at+b}.
\end{equation}
\end{proposition}

\begin{proof}
For $t=p_y/q_y$, direct expansion gives
\begin{equation}\label{eq:js-displayed-identity}
 q_yf_{a,b}(t)
 =ap_y\log p_y+bq_y\log q_y-m_y\log m_y.
\end{equation}
Summing proves the divergence identity.  Differentiation gives
$f_{a,b}'(t)=a\log[t/(at+b)]$ and \eqref{eq:js-second}.  Solving
$z=f_{a,b}'(t)$ yields \eqref{eq:js-conjugate}.  Finally, the matrix
logarithm is operator concave and $s\mapsto-\log(1-as)$ is operator
convex, which proves the operator assertions and \eqref{eq:js-score}.
\end{proof}

The following is not a priority claim: it is a direct specialization of
the prior Fang--Fawzi--Fawzi projectivization theorem, and the same
implication was explicitly recorded by Wang, Wang, and Chen
\cite{FangFawziFawzi2026,WangWangChen2026}.

\begin{corollary}[Binary projective sufficiency; prior corollary]\label{cor:shor}
Every finite-dimensional binary quantum ensemble has an
information-optimal von Neumann measurement.
\end{corollary}

\begin{proof}
Proposition~\ref{prop:js-lift} and Theorem~\ref{thm:fff} give the displayed
chain
\begin{equation}\label{eq:shor-two-line}
 \sup_{M\in\POVM}I_{\mathcal E}(M)
 =D_{f_{a,b}}^{\mathbb M}(\rho_0\Vert\rho_1)
 =D_{f_{a,b}}^{\mathrm{proj}}(\rho_0\Vert\rho_1).
\end{equation}
For faithful states the projective supremum is attained by compactness of
the flag manifold.  For arbitrary states, apply the faithful statement to
$\rho_x^{(\eps)}=(1-\eps)\rho_x+\eps\Id/d$.  Mutual information is jointly
continuous in $\eps$ and the PVM, uniformly on the compact flag manifold;
hence
\begin{equation}\label{eq:shor-limit-upper}
 \lim_{\eps\downarrow0}I_{\mathrm{acc}}(\mathcal E_\eps)
 =\max_{P\in\cP_d}I_{\mathcal E}(P).
\end{equation}
For every fixed POVM $M$,
$I_{\mathrm{acc}}(\mathcal E_\eps)\ge I_{\mathcal E_\eps}(M)
\to I_{\mathcal E}(M)$.  Taking the limit and then the supremum over finite
POVMs shows that the maximum on the right of
\eqref{eq:shor-limit-upper} equals $I_{\mathrm{acc}}(\mathcal E)$.
\end{proof}

\subsection{Historical note}

Fang, Fawzi, and Fawzi first posted their projectivization theorem on
February 11, 2025, and it subsequently appeared in \emph{IEEE Transactions
on Information Theory} \textbf{72}(3), 1751--1760
\cite{FangFawziFawzi2026}.  Their paper does not spell out the weighted
Jensen--Shannon specialization.  Shor's 2000 report records Fuchs--Peres
numerical evidence for binary ensembles while refuting Levitin's broader
conjecture; the surviving binary assertion was associated with Shor in
later literature \cite{Levitin1995,Shor2000}.  Thai and Dall'Arno still described the
arbitrary-dimensional assertion as open in their December 2025 preprint,
later published in 2026 \cite{ThaiDallArno2026}.  Wang, Wang, and Chen
posted a posterior-algebra proof and structural analysis on August 13,
2026 \cite{WangWangChen2026}; their preprint explicitly credits the same
Fang--Fawzi--Fawzi implication and makes no priority claim for projective
sufficiency.  Accordingly, we make no priority claim for the existence
statement.  The contribution developed below is the critical-point and
curvature theory.

%% file: sections/06-accessible.tex
\section{Binary Accessible-Information Landscapes}\label{sec:accessible}

Section~\ref{sec:shor} identified binary mutual information with the
measured $f_{a,b}$-divergence and verified the imported projectivization
hypotheses.  We now expose the associated posterior geometry and the
strict-ascent curvature supplied by the present critical-point theory.

\subsection{Posterior effect and concave lift}

For a PVM $P$, define the posterior probability of the first letter by
\begin{equation}\label{eq:posterior-values}
 r_j=\frac{ap_j}{ap_j+bq_j}.
\end{equation}
The score operator and the \emph{posterior effect} are
\begin{equation}\label{eq:posterior-effect}
 \Gamma_P=\sum_j\frac{p_j/q_j}{a(p_j/q_j)+b}P_j,
 \qquad R_P=a\Gamma_P=\sum_jr_jP_j.
\end{equation}
After the change of variables $R=a\Gamma$, the concave lift is
\begin{equation}\label{eq:posterior-functional}
 \Phi_{\mathcal E}(R)
 =a\Tr\rho_0\log\frac Ra
 +b\Tr\rho_1\log\frac{\Id-R}{b},
 \qquad 0<R<\Id.
\end{equation}
Consequently,
\begin{equation}\label{eq:posterior-variational}
 I_{\mathrm{acc}}(\mathcal E)=\sup_{0<R<\Id}\Phi_{\mathcal E}(R),
 \qquad I_{\mathcal E}(P)=\Phi_{\mathcal E}(R_P).
\end{equation}
The upper-semicontinuous support convention extends
\eqref{eq:posterior-functional} to boundary effects, assigning value
$-\infty$ when a logarithm is singular on the support of its weighting
state.

\subsection{Triangular atomic density}

\begin{proposition}[Triangular Green kernel]\label{prop:triangular-weight}
The generator $f_{a,b}$ is positively atomizable with density
\begin{equation}\label{eq:triangular-weight}
 w_{a,b}(\lambda)=
 \begin{cases}
 a\lambda,&0\le\lambda\le b,\\
 b(1-\lambda),&b\le\lambda\le1.
 \end{cases}
\end{equation}
More precisely,
\begin{equation}\label{eq:js-atomic}
 f_{a,b}(t)=\int_0^1w_{a,b}(\lambda)
 \frac{(t-1)^2}{(1-\lambda)t+\lambda}\,\dd\lambda.
\end{equation}
\end{proposition}

\begin{proof}
Both sides of \eqref{eq:js-atomic} vanish with their first derivative at
$t=1$.  Equations~\eqref{eq:atom-second} and
\eqref{eq:triangular-weight} give
\begin{align}
 2\!\int_0^1\!\frac{w_{a,b}(\lambda)\,\dd\lambda}
 {((1-\lambda)t+\lambda)^3}
 &={}2a\!\int_0^b\!\frac{\lambda\,\dd\lambda}
 {((1-\lambda)t+\lambda)^3}\notag\\
 &\quad+2b\!\int_b^1\!\frac{(1-\lambda)\,\dd\lambda}
 {((1-\lambda)t+\lambda)^3}\notag\\
 &=\frac{ab}{t(at+b)},
\end{align}
which is \eqref{eq:js-second}.  Appendix~\ref{app:triangular} records an
elementary antiderivative check.
\end{proof}

The density is the Green kernel for the Jensen gap at prior $b$.  It is
continuous, nonnegative, piecewise linear, and vanishes at both endpoints;
the prior fixes the apex.

\subsection{From projective existence to a landscape criterion}

\begin{theorem}[Benign binary-information landscape]\label{thm:binary-landscape}
Let $\rho_0,\rho_1$ be faithful finite-dimensional states and let $a,b>0$.
On the manifold of rank-one PVMs:
\begin{enumerate}
\item every local maximum of $I_{\mathcal E}$ is globally optimal over all
POVMs;
\item every second-order stationary PVM is globally optimal over all POVMs;
\item every nonglobal critical PVM has a two-level rotation for which
\begin{equation}\label{eq:binary-atomic-curvature}
 I_{\mathcal E}''(0)
 =8|a_{jk}|^2\left(\frac1{q_j}+\frac1{q_k}\right)
 \mathcal K_{a,b}(t_\alpha)>0,
\end{equation}
where $a_{jk}=\bra j(\rho_0-t_\alpha\rho_1)\ket k$ lies in a
degenerate posterior block and
\begin{equation}\label{eq:binary-kernel}
 \mathcal K_{a,b}(t)=\int_0^1
 \frac{w_{a,b}(\lambda)\,\dd\lambda}
 {((1-\lambda)t+\lambda)^3}.
\end{equation}
\end{enumerate}
\end{theorem}

\begin{proof}
Propositions~\ref{prop:js-lift} and \ref{prop:triangular-weight} permit
direct application of Theorem~\ref{thm:benign-score-block} and
Theorem~\ref{thm:atomic-curvature}.
\end{proof}

\begin{corollary}[Interior PVMs for nonfaithful states]
\label{cor:binary-interior-extension}
The critical-point classification used in
Theorem~\ref{thm:binary-landscape} remains valid for arbitrary
$\rho_0,\rho_1\ge0$ at any PVM for which every $p_j,q_j$ is positive.
\end{corollary}

\begin{proof}
The first- and second-variation arguments are local and require only the
positivity of the displayed outcome probabilities.  For the weighted
Jensen--Shannon generator, \eqref{eq:js-psi} parametrizes the full effective
domain of \eqref{eq:js-conjugate}; hence the original
Fang--Fawzi--Fawzi theorem supplies the POVM/projective equality for
arbitrary positive-semidefinite states.  The proof of
Theorem~\ref{thm:block-rigidity} then applies verbatim at the stated PVM.
\end{proof}

%% file: sections/07-qubit.tex
\section{Binary qubit ensembles and Keil's stationary-point conjecture}\label{sec:qubit}

\subsection{Bloch-plane reduction}

Using the standard Bloch representation and Pauli vector
\cite[Ch.~2]{NielsenChuang2010}, let $\cH=\CC^2$ and write
\begin{equation}\label{eq:bloch-states}
 \rho_x=\frac12(\Id+r_x\cdot\bm\sigma),
 \qquad r_x\in\RR^3,
 \qquad |r_x|\le1,
 \qquad x=0,1.
\end{equation}
A rank-one PVM is specified by an unoriented unit axis $n\in\mathbb S^2/\{\pm1\}$,
\begin{equation}\label{eq:qubit-pvm}
 P_\pm(n)=\frac12(\Id\pm n\cdot\bm\sigma).
\end{equation}
For the plus outcome,
\begin{equation}\label{eq:qubit-probs}
 p(n)=\frac{1+r_0\cdot n}{2},
 \qquad
 q(n)=\frac{1+r_1\cdot n}{2}.
\end{equation}
The minus probabilities are $1-p$ and $1-q$. Let
\begin{equation}
 \Pi=\spanop_{\RR}\{r_0,r_1\}.
\end{equation}
If $\Pi$ is one dimensional, fix any two-dimensional plane containing it. The \emph{state-plane projective circle} is the set of unoriented axes $n$ in that plane.

Let
\begin{equation}\label{eq:binary-entropy}
 h(u)=-u\log u-(1-u)\log(1-u).
\end{equation}
Then
\begin{equation}\label{eq:qubit-mutual}
 I(n)=h(ap(n)+bq(n))-ah(p(n))-bh(q(n)).
\end{equation}

\begin{lemma}[In-plane stationarity is full stationarity]\label{lem:plane-stationary}
Suppose all four conditional probabilities are positive. If $n$ is stationary for \eqref{eq:qubit-mutual} along the state-plane projective circle, then the corresponding PVM is stationary on the full qubit PVM manifold.
\end{lemma}

\begin{proof}
The function $I(n)$ depends on $n$ only through the two scalars $r_0\cdot n$ and $r_1\cdot n$. Its Euclidean gradient therefore lies in the state plane. At $n$ in that plane, the Riemannian gradient on the sphere is also in the plane and orthogonal to $n$. The in-plane tangent space is one dimensional and contains that vector. Vanishing of the in-plane derivative thus forces the full Riemannian gradient to vanish. Every out-of-plane tangent has zero derivative automatically.
\end{proof}

\subsection{Interior critical points}

\begin{proposition}[Optimal-or-independent interior dichotomy]\label{prop:qubit-interior}
Let $\rho_0\ne\rho_1$, and let $n$ be an in-plane stationary PVM for which all conditional probabilities are positive. Then exactly one of the following holds:
\begin{enumerate}
\item the PVM is globally information optimal over all POVMs;
\item $p(n)=q(n)$, and hence $I(n)=0$.
\end{enumerate}
\end{proposition}

\begin{proof}
By Lemma~\ref{lem:plane-stationary}, the PVM is a full projective critical point. Apply the score-block classification to the two outcomes. If the two posterior scores are distinct, both score blocks have rank one. Their residuals vanish automatically because
\begin{equation}
 P_\pm(\rho_0-t_\pm\rho_1)P_\pm=(p_\pm-t_\pm q_\pm)P_\pm=0.
\end{equation}
 Corollary~\ref{cor:binary-interior-extension} makes the PVM globally
 optimal.

If the scores coincide, injectivity of \eqref{eq:js-score} gives $t_+=t_-=:t$. Since $p_++p_-=q_++q_-=1$, one has $t=1$. Thus the two output distributions coincide and mutual information is zero.
The alternatives are mutually exclusive: nonidentical states with positive
priors admit a two-outcome measurement whose output distributions differ,
and hence whose mutual information is strictly positive.
\end{proof}

The zero-information condition has a unique solution on the projective circle when the states differ:
\begin{equation}\label{eq:independence-axis}
 p(n)=q(n)
 \quad\Longleftrightarrow\quad
 (r_0-r_1)\cdot n=0.
\end{equation}
In a two-dimensional state plane, a nonzero vector has exactly one orthogonal unoriented axis. Hence the independence minimum is unique.

\subsection{Strict concavity of the posterior-effect problem}

The remaining issue is uniqueness of the global maximum. We prove a statement that is useful beyond qubits.

For an ensemble $\mathcal E$ define the effective domain
\begin{equation}\label{eq:Phi-domain}
 \begin{aligned}
 \mathfrak D_{\mathcal E}=\{0\le R\le\Id:{}
 &\supp\rho_0\subseteq\supp R,\\[-2pt]
 &\supp\rho_1\subseteq\supp(\Id-R)\}.
 \end{aligned}
\end{equation}
On this domain, \eqref{eq:posterior-functional} is interpreted by compression to the relevant supports; outside it the value is $-\infty$.

\begin{lemma}[Strict concavity of the posterior functional]\label{lem:posterior-strict}
Let $\mathcal E=\{(a,\rho_0),(b,\rho_1)\}$ with $a,b>0$.  Assume
\begin{equation}\label{eq:joint-support}
 \supp\rho_0+\supp\rho_1=\cH.
\end{equation}
Then $\Phi_{\mathcal E}$ is strictly concave on its finite effective domain. Consequently it has at most one maximizer.
\end{lemma}

\begin{proof}
Let $R_0$ and $R_1$ be distinct points of the finite effective domain, put $H=R_1-R_0$, and set $R_s=(1-s)R_0+sR_1$ for $0<s<1$.  Positivity implies
\begin{equation}
 \begin{aligned}
 \ker R_s&=\ker R_0\cap\ker R_1,\\
 \ker(\Id-R_s)&=\ker(\Id-R_0)\cap\ker(\Id-R_1).
 \end{aligned}
\end{equation}
Moreover, $H$ vanishes on both common kernels.  We may therefore compress the first logarithm to $(\ker R_s)^\perp$ and the second to $(\ker(\Id-R_s))^\perp$; on those active spaces the corresponding operators are positive definite.  The second Fr\'echet derivative of the logarithm is
\begin{equation}\label{eq:log-second}
 \begin{gathered}
 B_t=(X+t\Id)^{-1},\\
 D^2\log_X[H,H]
 =-2\int_0^\infty B_tHB_tHB_t\,\dd t.
 \end{gathered}
\end{equation}
The operator under the integral before the minus sign is positive semidefinite, because with $B=(X+t\Id)^{-1}$,
\begin{equation}
 BHBHB=B^{1/2}(B^{1/2}HB^{1/2})^2B^{1/2}\ge0.
\end{equation}
Consequently the restriction $s\mapsto\Phi_{\mathcal E}(R_s)$ is concave.  If its second derivative vanished at some $s\in(0,1)$, both nonnegative trace integrals would vanish.  Multiplying their integrands by $t^3$ and taking $t\to\infty$ gives
\begin{equation}
 \Tr(\rho_0H^2)=0,
 \qquad
 \Tr(\rho_1H^2)=0.
\end{equation}
Thus $H$ annihilates both state supports.  Condition~\eqref{eq:joint-support} forces $H=0$, contrary to $R_0\ne R_1$.  Hence $\Phi_{\mathcal E}$ is strictly concave on its finite effective domain.
\end{proof}

\begin{proposition}[Unique information-optimal qubit PVM]\label{prop:unique-qubit-max}
Let $\mathcal E=\{(a,\rho_0),(b,\rho_1)\}$ with $a,b>0$.  If
$\rho_0\ne\rho_1$ are qubit states, then the posterior functional has a
unique maximizer $R_*$. This maximizer is nonscalar, and its spectral PVM
is the unique information-optimal PVM, up to outcome relabeling.
\end{proposition}

\begin{proof}
Two nonidentical qubit states have joint support $\CC^2$: the only way their supports could span a one-dimensional space is for both to be the same pure state. Lemma~\ref{lem:posterior-strict} gives uniqueness of $R_*$. Existence follows from compactness of the effect interval and upper semicontinuity with the support convention.

Suppose $R_*=c\Id$. Finiteness forces $0<c<1$. The first-order condition at this interior point is
\begin{equation}
 \frac{a}{c}\rho_0-\frac{b}{1-c}\rho_1=0.
\end{equation}
Taking traces gives $c=a$, after which the matrix equality gives $\rho_0=\rho_1$, a contradiction. Hence $R_*$ is nonscalar.

 Corollary~\ref{cor:shor} supplies at least one optimal PVM $P$.
 Equation~\eqref{eq:posterior-variational} then gives
 $\Phi_{\mathcal E}(R_P)=I_{\mathrm{acc}}(\mathcal E)$, so uniqueness of the
 maximizer implies $R_P=R_*$.  A nonscalar Hermitian $2\times2$ matrix has a
 unique spectral PVM up to swapping its eigenvalues.  Therefore the optimal
 PVM is unique.
\end{proof}

\subsection{Singular boundary measurements}

If one state is pure, an in-plane axis can make one conditional probability vanish. The projective objective remains differentiable as a function of the angle, but the formula using $f''(t)$ does not apply directly. The following lemma supplies the missing classification; its endpoint estimates are expanded in Appendix~\ref{app:boundary}.

\begin{lemma}[Boundary stationary points]\label{lem:boundary-stationary}
Let $\rho_0\ne\rho_1$ be qubit states, and let an in-plane PVM have at least one zero conditional probability. If it is stationary along the projective circle, then the two states commute and the PVM is their common eigenbasis. In particular, it is globally information optimal.
\end{lemma}

\begin{proof}
It is enough to treat $p(0)=0$; the other cases follow by exchanging states or outcomes. Then $\rho_0$ is pure and the measurement axis is opposite its Bloch vector. Along an angular coordinate $\theta$ in the state plane,
\begin{equation}
 p(0)=p'(0)=0,
 \qquad p(\theta)=O(\theta^2).
\end{equation}
The singular entropy term satisfies $h'(p(\theta))p'(\theta)\to0$. If $0<q(0)<1$, differentiation of \eqref{eq:qubit-mutual} gives
\begin{equation}\label{eq:boundary-derivative}
 I'(0)=bq'(0)\bigl[h'(bq(0))-h'(q(0))\bigr].
\end{equation}
Since $h'$ is strictly decreasing and $0<b<1$, the bracket is nonzero. Stationarity therefore requires $q'(0)=0$. In the state plane, this means that $r_1$ is parallel to the pure Bloch vector $r_0$, so the states commute and the measurement is their common eigenbasis. Measuring that basis is sufficient for the resulting classical pair and is globally optimal.

If $q(0)=0$, both states are the same pure state, contrary to the hypothesis. If $q(0)=1$, the states are orthogonal pure states and the same measurement perfectly distinguishes them. These cases complete the proof.
\end{proof}

\subsection{Keil's conjecture}

We can now state the complete qubit result. A ``stationary measurement'' is an unoriented projective axis; swapping the two outcomes does not produce a second point.

\begin{theorem}[Complete binary-qubit stationary landscape]\label{thm:keil}
Let $\mathcal E=\{(a,\rho_0),(b,\rho_1)\}$ be a binary qubit ensemble with $a,b>0$ and $\rho_0\ne\rho_1$. On the projective circle of measurement axes in the Bloch plane of the states, the mutual information has exactly two stationary measurements:
\begin{enumerate}
\item the unique axis satisfying $(r_0-r_1)\cdot n=0$, at which $I=0$ and which is the global minimum;
\item the unique spectral PVM of $R_*$, at which $I=I_{\mathrm{acc}}(\mathcal E)$ and which is the global maximum.
\end{enumerate}
If $\rho_0=\rho_1$, then $I$ vanishes identically and every measurement is stationary.
\end{theorem}

\begin{proof}
Interior stationary points are classified by Proposition~\ref{prop:qubit-interior}; boundary stationary points by Lemma~\ref{lem:boundary-stationary}. Thus every stationary point for nonidentical states is either a global maximum or the independence minimum. Equation~\eqref{eq:independence-axis} makes the minimum unique, and Proposition~\ref{prop:unique-qubit-max} makes the maximum unique. Both extrema exist on the compact projective circle, so there are exactly two.
\end{proof}

\begin{remark}[Relation to Keil's results]
Keil's 2008 preprint, published in 2024, proves that every local maximum of
binary-qubit mutual information over POVMs is a von Neumann measurement and
develops the associated scalar optimality equation \cite{Keil2024}.  The
exact two-stationary-point statement appears separately as Conjecture~2 on
p.~77 of his 2009 doctoral thesis \cite{KeilThesis2009}.
Theorem~\ref{thm:keil} proves that conjecture by embedding the projective
problem into the concave posterior-effect geometry and resolving the only
possible degenerate score block.
\end{remark}

%% file: sections/08-thai.tex
\section{Quasi-concavity, pseudo-concavity, and bisection}\label{sec:thai}

\subsection{The Thai--Dall'Arno coordinate}

Thai and Dall'Arno's December 2025 preprint, subsequently published in
2026, gives a one-parameter description of the projective circle relevant
to a binary qubit ensemble \cite{ThaiDallArno2026}.  We recall only the
ingredients needed here.

Let $\Delta=\rho_0-\rho_1$ and, when $\Delta\ne0$, set
\begin{equation}\label{eq:thai-mu}
 \mu=\frac{\Tr[\rho_1(\rho_1-\rho_0)]}{\Tr[(\rho_0-\rho_1)^2]},
 \qquad
 \omega_\mu=\mu\rho_0+(1-\mu)\rho_1.
\end{equation}
This choice makes $\Tr(\omega_\mu\Delta)=0$. Consider the Hermitian pencil
\begin{equation}\label{eq:thai-pencil}
 H(\lambda)=\lambda\omega_\mu-\Delta,
 \qquad \lambda\in\overline{\RR}:=\RR\cup\{\pm\infty\}.
\end{equation}
Whenever $H(\lambda)$ has two distinct eigenvalues, its two rank-one spectral
projectors define a PVM.  On the interval where $H(\lambda)$ is indefinite,
these projectors agree with the positive- and negative-spectral-subspace
measurement used in \cite{ThaiDallArno2026} to generate the
state-dependent extremal boundary of the testing region.

\begin{proposition}[Thai--Dall'Arno circle coordinate; imported]
\label{prop:thai-coordinate}
If $\omega_\mu\not\propto\Id$, the extended parameter line in
\eqref{eq:thai-pencil} traverses the state-plane projective circle once;
the two infinities define the same PVM, namely the eigenbasis of
$\omega_\mu$ \cite{ThaiDallArno2026}.
\end{proposition}

Since $\Tr[\omega_\mu(\rho_0-\rho_1)]=0$, that limiting PVM has identical
output distributions.  The scalar case $\omega_\mu\propto\Id$ occurs in
the collinear Bloch-vector geometry and is handled directly as a commuting
degeneracy.

Thai and Dall'Arno reduce the search for the maximizer to
\begin{equation}\label{eq:thai-interval}
 [-\lambda_*,\lambda_*],
 \qquad
 \begin{gathered}
 \Delta=\rho_0-\rho_1,\\
 \Xi=\Tr\Delta^2-\Tr(\rho_0^2)\Tr(\rho_1^2)
 +(\Tr\rho_0\rho_1)^2,\\
 \lambda_*=\Tr\Delta^2/\sqrt{\Xi}.
 \end{gathered}
\end{equation}
We denote the induced mutual information by $I_{\mathrm{TD}}(\lambda)$.

The exact formula for the coordinate is prior work. The contribution here is to remove the conjectural shape assumption under which its algorithmic use was justified.

\subsection{Unimodality from the stationary-point theorem}

\begin{lemma}[Two critical points imply strict unimodality]\label{lem:circle-unimodal}
Let $F$ be a continuously differentiable function on a circle with exactly two critical points, one a unique global minimum and one a unique global maximum. On each of the two open arcs joining the minimum to the maximum, $F$ is strictly monotone.
\end{lemma}

\begin{proof}
The derivative has no zero on either open arc. By continuity it has constant sign there. The endpoint values determine the sign.
\end{proof}

Applied to Theorem~\ref{thm:keil}, this gives one increasing and one
decreasing arc.  When $\omega_\mu\not\propto\Id$, the extended coordinate
\eqref{eq:thai-pencil} cuts the circle at the minimum, so it turns those arcs
into the two sides of the unique finite maximizer.  When
$\omega_\mu\propto\Id$, the two states commute and the pencil has a fixed
eigenbasis; the coordinate objective is constant wherever the pencil is
nondegenerate and does not parametrize the circle.

\begin{theorem}[Thai--Dall'Arno shape conjectures]\label{thm:thai-conjectures}
For every nonidentical binary qubit ensemble with positive priors and
$\omega_\mu\not\propto\Id$, the function $I_{\mathrm{TD}}$ has a unique
finite maximizer $\lambda_{\max}$ and satisfies
\begin{equation}\label{eq:thai-sign}
 I_{\mathrm{TD}}'(\lambda)>0\quad(\lambda<\lambda_{\max}),
 \qquad
 I_{\mathrm{TD}}'(\lambda)<0\quad(\lambda>\lambda_{\max})
\end{equation}
whenever the derivative is represented in a nonsingular chart. Consequently:
\begin{enumerate}
\item $I_{\mathrm{TD}}$ is quasi-concave on the extended line $[-\infty,+\infty]$ with the two endpoints identified;
\item its restriction to the finite search interval \eqref{eq:thai-interval} is pseudo-concave;
\item every stationary point in the finite interval is globally maximizing.
\end{enumerate}
If $\omega_\mu\propto\Id$, the states commute, $I_{\mathrm{TD}}$ is constant
on every nondegenerate branch, and the non-strict forms of quasi-concavity
and pseudo-concavity hold.
\end{theorem}

\begin{proof}
Under the nonscalar hypothesis, Theorem~\ref{thm:keil} supplies exactly one
minimum and one maximum on the projective circle. Lemma~\ref{lem:circle-unimodal}
gives strict monotonicity on the two complementary arcs.
Proposition~\ref{prop:thai-coordinate} places the minimum at the identified
endpoints and covers the rest of the circle once, so \eqref{eq:thai-sign}
follows.  In the
scalar case, $H(\lambda)=c\lambda\Id-\Delta$ has the $\lambda$-independent
eigenbasis of $\Delta$, which proves the stated constant-function conclusion.

A one-dimensional differentiable function with the sign pattern \eqref{eq:thai-sign} is quasi-concave because all upper level sets are intervals. It is pseudo-concave because, if $I(y)>I(x)$, then either $x$ and $y$ lie on one monotone side of the maximizer, or they lie on opposite sides; in every case $I'(x)(y-x)>0$. The last assertion is immediate from the uniqueness of the stationary point.
\end{proof}

For identical states, mutual information vanishes for every PVM.  This
trivial case is independent of the Thai--Dall'Arno coordinate, whose
definition above assumes $\Delta\ne0$.

Theorem~\ref{thm:thai-conjectures} proves the strict forms of Conjectures 1
and 2 of \cite{ThaiDallArno2026} in the nonscalar coordinate regime and their
non-strict shape conclusions in the scalar commuting degeneracy. Their
Conjecture 1 is also presented there as a reformulation of Keil's
stationary-point conjecture. The equivalence is transparent in the nonscalar
regime: on a circle with a known unique minimum, the absence of any additional
stationary point is exactly the strict unimodality needed for quasi-concavity.

\subsection{Exact-arithmetic bisection}

Thai and Dall'Arno propose a derivative-bisection algorithm whose guarantee
is conditional on their pseudo-concavity conjecture.  The preceding theorem
establishes that hypothesis in the nonscalar regime used by
Corollary~\ref{cor:bisection}.

\begin{corollary}[Bisection guarantee]\label{cor:bisection}
Assume $\omega_\mu\not\propto\Id$, exact evaluation of
$I_{\mathrm{TD}}'$, and an initial bracket
$[L,U]\subseteq[-\lambda_*,\lambda_*]$ with
\begin{equation}
 I_{\mathrm{TD}}'(L)\ge0,
 \qquad
 I_{\mathrm{TD}}'(U)\le0.
\end{equation}
Derivative bisection returns an interval containing the unique maximizer after every iteration. After $N$ iterations its width is at most
\begin{equation}\label{eq:bisection-width}
 2\lambda_*\,2^{-N}.
\end{equation}
In particular, parameter precision $\delta$ requires at most
\begin{equation}\label{eq:bisection-count}
 \left\lceil\log_2\frac{2\lambda_*}{\delta}\right\rceil
\end{equation}
derivative evaluations beyond initialization.
\end{corollary}

\begin{proof}
The derivative sign changes exactly once by \eqref{eq:thai-sign}. Standard bisection preserves a bracket around that sign change and halves its width at each step.
\end{proof}

\begin{remark}[Numerical versus mathematical guarantees]
Corollary~\ref{cor:bisection} is an exact-arithmetic statement. A rigorous floating-point implementation must control eigenvector branch choices near degeneracy and enclose the derivative sign. Interval arithmetic gives a direct route to validated sign decisions. The theorem does not by itself bound the loss in mutual information from a given parameter error; that requires a local curvature or Lipschitz bound.
\end{remark}

%% file: sections/07-consequences.tex
\section{Dynamic and order-theoretic consequences of the posterior lift}
\label{sec:lift-consequences}

The posterior effect in Section~\ref{sec:accessible} carries information
beyond the projective landscape.  Allowing the prior to vary gives a
jointly concave matrix perspective and another proof of the prior-value
concavity used in Shor's 2002/2004 adaptive-capacity argument
\cite{ShorAdaptive2004}; Xiao's 2021 report had already claimed and argued
that concavity and the resulting capacity equality \cite{Xiao2021}.  At
the rare-prior boundary, the normalized triangular density converges to
the relative-entropy density.  This gives a direct measurement-channel
proof of the known mutual-information/relative-entropy order equivalence
and, using the published Example~16 of
\cite{TeixidoSchindlerSafranek2025}, an explicit finite witness against
the conjectured observational-entropy equivalence.

\subsection{Joint prior--effect concavity}

For states $\rho_0,\rho_1$ and $0<a<1$, write
\begin{equation}\label{eq:prior-value}
 \mathcal A_{\rho_0,\rho_1}(a)
 =\sup_{M\in\POVM} I_{\{(a,\rho_0),(1-a,\rho_1)\}}(M).
\end{equation}
The fixed-prior equality below is equivalent to the posterior program used
by Wang, Wang, and Chen \cite{WangWangChen2026} and follows from the same
operator-Jensen mechanism as Theorem~\ref{thm:fff}.  Xiao's 2021 report
states and argues the resulting prior-value concavity \cite{Xiao2021}.
The additional structural statement emphasized here is joint concavity in
the prior and effect.

\begin{theorem}[Posterior lift and joint prior--effect concavity]
\label{thm:prior-concavity}
For arbitrary finite-dimensional states $\rho_0,\rho_1$,
\begin{equation}\label{eq:joint-posterior-program}
 \mathcal A_{\rho_0,\rho_1}(a)
 =\max_{0\preceq R\preceq\Id}\Phi(a,R),
\end{equation}
where
\begin{equation}\label{eq:joint-posterior-functional}
 \Phi(a,R)=a\Tr\rho_0\log\frac{R}{a}
 +(1-a)\Tr\rho_1\log\frac{\Id-R}{1-a}.
\end{equation}
Boundary values use the upper-semicontinuous support convention.  The map
$(a,R)\mapsto\Phi(a,R)$ is jointly concave, and
$a\mapsto\mathcal A_{\rho_0,\rho_1}(a)$ is concave on $[0,1]$.
\end{theorem}

\begin{proof}
Fix a POVM $M=(M_y)_y$.  Put
\begin{equation}\label{eq:posterior-data-variable-prior}
 u_y=\Tr(\rho_0M_y),\quad v_y=\Tr(\rho_1M_y),\quad
 m_y=au_y+(1-a)v_y,
\end{equation}
omit outcomes with $m_y=0$, and define
\begin{equation}\label{eq:posterior-effect-variable-prior}
 r_y=\frac{au_y}{m_y},\qquad R_M=\sum_y r_yM_y.
\end{equation}
If some $r_y$ equals $0$ or $1$, first replace it by
$r_{y,\delta}=(1-2\delta)r_y+\delta$, apply the following inequalities,
trace against the corresponding states, and let $\delta\downarrow0$.  Thus
the operator-Jensen step is understood in the extended support sense.
Operator concavity of the logarithm gives
\begin{align}
 \sum_y(\log r_y)M_y&\preceq\log R_M,\label{eq:posterior-log-one}\\
 \sum_y\log(1-r_y)M_y&\preceq\log(\Id-R_M).
 \label{eq:posterior-log-two}
\end{align}
After tracing against $a\rho_0$ and $(1-a)\rho_1$, respectively, these
inequalities yield $I(M)\leq\Phi(a,R_M)$.

Conversely, let $R=\sum_js_jQ_j$ be its spectral decomposition.  For its
spectral PVM define
\begin{equation}\label{eq:spectral-pvm-data}
 \begin{aligned}
 u_j&=\Tr(\rho_0Q_j),& v_j&=\Tr(\rho_1Q_j),\\
 m_j&=au_j+(1-a)v_j,& \widehat r_j&=au_j/m_j.
 \end{aligned}
\end{equation}
Terms with $m_j=0$ are omitted.  Direct subtraction gives
\begin{equation}\label{eq:spectral-posterior-remainder}
 I(Q)-\Phi(a,R)
 =\sum_jm_jD_{\rm bin}(\widehat r_j\Vert s_j)\geq0,
\end{equation}
with the standard extended binary-relative-entropy convention.  Thus every
effect is dominated by its spectral PVM, proving
\eqref{eq:joint-posterior-program}.  The effect interval is compact and the
closed functional is upper semicontinuous, so the maximum is attained.

For $(a_i,R_i)$, $i=1,2$, and $0<\theta<1$, put
$a=\theta a_1+(1-\theta)a_2$ and
 $R=\theta R_1+(1-\theta)R_2$.  The operator perspective inequality
 \cite{Effros2009}
\begin{align}
 a\log\frac Ra
 &\succeq \theta a_1\log\frac{R_1}{a_1}
 +(1-\theta)a_2\log\frac{R_2}{a_2}
 \label{eq:operator-perspective-prior}
\end{align}
follows from operator concavity of $\log$; the same argument applies to
$(1-a,\Id-R)$.  Hence $\Phi$ is jointly concave in the interior.  The
regularization
\begin{equation}\label{eq:closed-perspective-regularization}
 a_\delta=(1-2\delta)a+\delta,\qquad
 R_\delta=(1-2\delta)R+\delta\Id
\end{equation}
and upper-semicontinuous closure give the boundary statement.  Taking
convex combinations of maximizing effects proves concavity of
$\mathcal A$ on $(0,1)$.  Finally,
\begin{equation}\label{eq:prior-endpoint-bound}
 0\leq\mathcal A_{\rho_0,\rho_1}(a)\leq h(a)
\end{equation}
with $h$ the binary entropy, so the result extends continuously to both
endpoints.
\end{proof}

\subsection{Adaptive local decoding}

Let $W$ be the binary classical--quantum channel
$0\mapsto\rho_0$, $1\mapsto\rho_1$.  In Shor's notation
\cite{ShorAdaptive2004}, its separate-measurement value is
\begin{equation}\label{eq:separate-capacity}
 C_{1,1}(W)=\max_{0\leq a\leq1}
 \mathcal A_{\rho_0,\rho_1}(a).
\end{equation}
An adaptive local decoder receives a product codeword and applies a finite
measurement tree whose instruments act on one tensor factor at a time.
The selected factor and instrument may depend on the previous transcript,
and a factor may be revisited.  If $U$ indexes a deterministic binary
codeword $X^n(U)$ and $H_T$ is the terminal transcript, define
\begin{equation}\label{eq:adaptive-capacity-definition}
 C_{1,A}(W)=\sup_{n,U,X^n,\mathrm{decoder}}
 \frac1n I(U;H_T).
\end{equation}

\begin{lemma}[One-system continuation]
\label{lem:continuation}
Let $X\in\{0,1\}$ have prior $a$ and be encoded as $\tau_0,\tau_1$.
Apply an instrument with outcome $Y$, posterior $a_y$, and normalized
residual states $\tau_{x|y}$.  Then
\begin{equation}\label{eq:continuation-bound}
 I(X;Y)+\sum_y\Pr(Y=y)
 \mathcal A_{\tau_{0|y},\tau_{1|y}}(a_y)
 \leq\mathcal A_{\tau_0,\tau_1}(a).
\end{equation}
Zero-probability branches are omitted.
If a retained branch is impossible under one input letter, choose the
corresponding residual state arbitrarily; its posterior is an endpoint and
the associated accessible information is zero.
\end{lemma}

\begin{proof}
On each nonzero branch, follow the instrument by a POVM approaching the
branch accessible information.  The composed experiment is one POVM on
the original system.  The mutual-information chain rule
\cite[Ch.~2]{CoverThomas2006} identifies its value with the left-hand side
of \eqref{eq:continuation-bound}, up to the arbitrarily small optimization
error.
\end{proof}

Shor posed the equality below as Conjecture~2 in his 2002 preprint,
published in 2004 \cite{ShorAdaptive2004}, and Xiao's 2021 public report
states and presents an argument for both the required prior concavity and
the equality \cite{Xiao2021}.  We give a different proof within the
posterior-lift framework.

\begin{theorem}[Adaptive binary capacity: proof in the present framework]
\label{thm:adaptive-capacity}
For every finite-dimensional binary classical--quantum channel,
\begin{equation}\label{eq:adaptive-capacity-equality}
 C_{1,A}(W)=C_{1,1}(W).
\end{equation}
\end{theorem}

\begin{proof}
Conditioned on a transcript $h$, factor $i$ remains a binary experiment
with posterior
\begin{equation}\label{eq:factor-posterior}
 a_i(h)=\Pr(X_i=0\mid H_t=h)
\end{equation}
and residual state pair $\tau_{i,0}^h,\tau_{i,1}^h$.  At an endpoint
posterior, the residual state for the impossible letter is chosen
arbitrarily; $V_i(h)=0$, independently of that choice.  Define
\begin{align}
 V_i(h)&=\mathcal A_{\tau_{i,0}^h,\tau_{i,1}^h}(a_i(h)),
 \label{eq:factor-value}\\
 \mathcal V_t&=I(U;H_t)+
 \mathbb E\!\left[\sum_{i=1}^nV_i(H_t)\right].
 \label{eq:transcript-potential}
\end{align}
Indeed, at fixed $h$ the previously selected operations on factor $i$
compose to one completely positive map fixed by that transcript.  Applied
to a product codeword, its normalized residual state depends on $U$ only
through $X_i$.  A subsequent operation on another factor changes neither
this map nor this conditional state pair; transcript correlations update
only the binary prior.
Suppose the next instrument, fixed after conditioning on $H_t=h$, acts on
factor $i$ and produces $Y$.  Since $U$ determines $X_i$ and the outcome
law depends on $U$ through $X_i$,
\begin{equation}\label{eq:local-information-identity}
 I(U;Y\mid H_t=h)=I(X_i;Y\mid H_t=h).
\end{equation}
Lemma~\ref{lem:continuation} absorbs this immediate information gain and
the expected new value of factor $i$.  For every $j\ne i$, the residual
pair is unchanged and only its posterior is updated.  Theorem
\ref{thm:prior-concavity} and
$\mathbb E[a_j(H_t,Y)\mid H_t=h]=a_j(h)$ give
\begin{equation}\label{eq:side-information-concavity}
 \mathbb E[V_j(H_t,Y)\mid H_t=h]\leq V_j(h).
\end{equation}
Adding the inequalities proves $\mathcal V_{t+1}\leq\mathcal V_t$.
Finite transcript-dependent stopping is handled on the finite measurement
tree, equivalently by padding terminated branches with null operations.

Initially,
\begin{equation}\label{eq:initial-potential-bound}
 \mathcal V_0=\sum_{i=1}^n
 \mathcal A_{\rho_0,\rho_1}(a_i)\leq nC_{1,1}(W).
\end{equation}
All residual values are nonnegative, so
\begin{equation}\label{eq:terminal-information-bound}
 I(U;H_T)\leq\mathcal V_T\leq\mathcal V_0\leq nC_{1,1}(W).
\end{equation}
Taking the supremum in \eqref{eq:adaptive-capacity-definition} gives the
upper bound.  A one-system nonadaptive decoder is an allowed adaptive
protocol, giving the reverse inequality.
\end{proof}

The derivation above supplies a prior-uniform operator-perspective proof and
an explicit transcript potential.  We make no first-proof claim for either
prior concavity or \eqref{eq:adaptive-capacity-equality}.

\subsection{The rare-prior boundary and measurement orders}

For finite POVMs $M,N$, let $M\succeq_{\rm MI}N$ mean
$I(U;Y_M)\geq I(U;Y_N)$ for every finite ensemble.  Let
$M\succeq_{\rm KL}N$ mean
\begin{equation}\label{eq:measurement-kl-order}
 D(M(\rho)\Vert M(\sigma))
 \geq D(N(\rho)\Vert N(\sigma))
\end{equation}
for every state pair.  Observational entropy for
$M=(M_y)_y$, with zero effects omitted, is
\cite{SafranekDeutschAguirre2019,BuscemiSchindlerSafranek2023}
\begin{equation}\label{eq:observational-entropy}
 S_M(\rho)=-\sum_y\Tr(\rho M_y)
 \log\frac{\Tr(\rho M_y)}{\Tr M_y}.
\end{equation}
Write $M\succeq_{\rm OE}N$ when
$S_M(\rho)\leq S_N(\rho)$ for every state.

The equivalence below is the classical less-noisy/relative-entropy
characterization specialized to measurement channels
\cite{KornerMarton1977,MakurPolyanskiy2018}; the final implication is the
observational-entropy specialization also used in
\cite{TeixidoSchindlerSafranek2025}.  We include a direct proof for the
present quantum-input measurement setting.

\begin{theorem}[Measurement-channel form of the less-noisy order]
\label{thm:measurement-orders}
For finite-dimensional POVMs,
\begin{equation}\label{eq:measurement-order-hierarchy}
 M\succeq_{\rm MI}N
 \quad\Longleftrightarrow\quad
 M\succeq_{\rm KL}N
 \quad\Longrightarrow\quad
 M\succeq_{\rm OE}N.
\end{equation}
\end{theorem}

\begin{proof}
If $M\succeq_{\rm KL}N$ and
$\bar\rho=\sum_u\pi_u\rho_u$, then
\begin{equation}\label{eq:mutual-relative-decomposition}
 I(U;Y_M)=\sum_u\pi_u
 D(M(\rho_u)\Vert M(\bar\rho)).
\end{equation}
Termwise comparison proves $M\succeq_{\rm MI}N$.

For the converse, apply the all-ensemble inequality to
$\{(\eps,\rho),(1-\eps,\sigma)\}$.  If
$P=M(\rho)$, $Q=M(\sigma)$, and
$R_\eps=\eps P+(1-\eps)Q$, its information is
\begin{align}
 J_M(\eps)
 &=\eps D(P\Vert R_\eps)
 +(1-\eps)D(Q\Vert R_\eps),\label{eq:rare-prior-js}\\
 \lim_{\eps\downarrow0}\frac{J_M(\eps)}{\eps}
 &=D(P\Vert Q).
 \label{eq:rare-prior-relative-limit}
\end{align}
When $P\ll Q$, the limit follows from
\begin{equation}\label{eq:rare-prior-identity}
 J_M(\eps)=\eps D(P\Vert Q)-D(R_\eps\Vert Q).
\end{equation}
Because the alphabet is finite and $P\ll Q$, a Taylor expansion on
$\supp Q$ gives $D(R_\eps\Vert Q)=O(\eps^2)$.  Hence
\eqref{eq:rare-prior-identity} yields
\eqref{eq:rare-prior-relative-limit}.
If $P\not\ll Q$, a coordinate with $P_y>0=Q_y$ contributes
$\eps P_y\log(1/\eps)+O(\eps)$, giving the same extended-value limit.
Thus $J_M(\eps)\geq J_N(\eps)$ for every $\eps$ implies
\eqref{eq:measurement-kl-order}.  The extended cases use the convention
$+\infty\geq x$; a finite left limit and infinite right limit would
contradict the inequalities for small $\eps$.

Finally,
\begin{equation}\label{eq:oe-relative-identity}
 S_M(\rho)=\log d-D(M(\rho)\Vert M(\Id/d)).
\end{equation}
Setting $\sigma=\Id/d$ proves the last implication.  The classical
less-noisy characterization behind the equivalence is reviewed in
\cite{KornerMarton1977,MakurPolyanskiy2018}; Buscemi gives its quantum
channel formulation \cite{Buscemi2016}.
\end{proof}

The converse implication from observational entropy was conjectured in the
authors' October 2023 preprint, whose published version appeared in the
2025 volume of \emph{Physica Scripta}
\cite{TeixidoSchindlerSafranek2025}.  Their Example~16 already provides the
measurement pair, the value $\lambda=1/64$, the observational-entropy
comparison, and failure of the corresponding relative-entropy order.  The
following argument fixes a concrete rare prior and writes an explicit
finite ensemble witness; it is not a claim of a new measurement-pair
separation.

We use the standard binary entropy \cite[Ch.~2]{CoverThomas2006} in the
elementary estimate below.
\begin{lemma}[Binary-entropy bounds]\label{lem:binary-entropy-bounds}
For $|t|\leq1/2$, with entropy measured in bits,
\begin{equation}\label{eq:binary-entropy-bounds}
 \frac{2t^2}{\ln2}
 \leq1-h_2\!\left(\frac12+t\right)
 \leq4t^2.
\end{equation}
\end{lemma}

\begin{proof}
For the lower bound, subtract $2t^2/\ln2$ from the middle expression;
the result and its first derivative vanish at zero, while its second
derivative is
\begin{equation}\label{eq:entropy-lower-second}
 \frac1{\ln2}\left(\frac1{1/4-t^2}-4\right)\geq0.
\end{equation}
For the upper bound, set $u=2|t|$ and subtract
$1-h_2((1+u)/2)$ from $u^2$.  The difference vanishes at $0$ and $1$;
its derivative increases once and then decreases to $-\infty$, because
\begin{equation}\label{eq:entropy-upper-second}
 \frac{\dd^2}{\dd u^2}
 \left[u^2-1+h_2\!\left(\frac{1+u}{2}\right)\right]
 =2-\frac1{(1-u^2)\ln2}
\end{equation}
is strictly decreasing.  Hence the difference is nonnegative.
\end{proof}

\begin{theorem}[Finite-ensemble witness for the observational-entropy reversal]
\label{thm:entropic-order-reversal}
On a qubit, define
\begin{align}
 M&=\left(\frac12\proj{0}+\frac14\Id,
           \frac12\proj{1}+\frac14\Id\right),
 \label{eq:noisy-binary-measurement}\\
 N_\lambda&=\left(\lambda\proj{0},\lambda\proj{1},
                   (1-\lambda)\Id\right),
 \qquad \lambda=\frac1{64}.
 \label{eq:erasure-measurement}
\end{align}
Then $M\succeq_{\rm OE}N_\lambda$, while the binary ensemble
\begin{equation}\label{eq:explicit-rare-ensemble}
 \{(\eps,\proj{0}),(1-\eps,\proj{1})\},
 \qquad \eps=2^{-64},
\end{equation}
satisfies
\begin{equation}\label{eq:explicit-information-reversal}
 I(U;Y_{N_\lambda})>I(U;Y_M).
\end{equation}
Consequently observational-entropy order and all-ensemble
mutual-information order are inequivalent.
\end{theorem}

\begin{proof}
Let $x=\bra0\rho\ket0$ and use bits.  Since the two effects of $M$ have
unit trace, while the volumes of $N_\lambda$ are
$(\lambda,\lambda,2(1-\lambda))$,
\begin{align}
 S_M(\rho)&=h_2\!\left(\frac14+\frac x2\right),
 \label{eq:noisy-measurement-oe}\\
 S_{N_\lambda}(\rho)&=1-\lambda+\lambda h_2(x).
 \label{eq:erasure-measurement-oe}
\end{align}
Writing $\delta=x-1/2$ and using
Lemma~\ref{lem:binary-entropy-bounds} gives
\begin{align}
 S_{N_\lambda}(\rho)-S_M(\rho)
 &\geq\left(\frac1{2\ln2}-4\lambda\right)\delta^2>0
 \label{eq:oe-positive-gap}
\end{align}
unless $\delta=0$, where equality holds.

For the ensemble \eqref{eq:explicit-rare-ensemble}, $N_\lambda$ reveals
the letter perfectly with probability $\lambda$ and otherwise returns an
input-independent outcome.  Therefore
\begin{equation}\label{eq:erasure-information-lower}
 I(U;Y_{N_\lambda})=\lambda h_2(\eps)
 >\lambda\eps\log_2(1/\eps)=\eps.
\end{equation}
The conditional distributions generated by $M$ are
$P=(3/4,1/4)$ and $Q=(1/4,3/4)$.  Applying
\eqref{eq:rare-prior-identity} in bits yields
\begin{equation}\label{eq:noisy-information-upper}
 I(U;Y_M)\leq\eps D_2(P\Vert Q)
 =\frac\eps2\log_2 3<\eps.
\end{equation}
This proves the strict reversal.
\end{proof}

\subsection{Atomic convergence at the boundary}

Return to natural logarithms.  The preceding order theorem is the boundary
form of the weighted Jensen--Shannon family in
Proposition~\ref{prop:triangular-weight}, whose classical provenance goes
back to Lin \cite{Lin1991} and whose positive $f$-decomposition follows the
resolvent framework of \cite{SalazarAtomic2026}.

\begin{proposition}[Rare-prior atomic limit]
\label{prop:rare-prior-atomic-limit}
For $f_{\eps,1-\eps}$ from \eqref{eq:js-generator} and its triangular
density $w_{\eps,1-\eps}$ from \eqref{eq:triangular-weight},
\begin{align}
 \lim_{\eps\downarrow0}\frac{f_{\eps,1-\eps}(t)}{\eps}
 &=t\log t-t+1,\qquad t>0,
 \label{eq:rare-generator-limit}\\
 \left\|\frac{w_{\eps,1-\eps}(\lambda)}{\eps}-\lambda
 \right\|_{L^1(0,1)}
 &=\frac\eps2.
 \label{eq:rare-atomic-l1}
\end{align}
Thus the complete positive atomic density converges in $L^1$ to the
relative-entropy density $\lambda$.
\end{proposition}

\begin{proof}
The first limit follows by differentiating
$(1+\eps(t-1))\log(1+\eps(t-1))$ at $\eps=0$.  On
$0\leq\lambda\leq1-\eps$, the normalized triangular density equals
$\lambda$.  On the remaining interval, direct integration gives
\begin{equation}\label{eq:rare-atomic-endpoint-integral}
 \int_{1-\eps}^1
 \left[\lambda-\frac{1-\eps}{\eps}(1-\lambda)\right]\dd\lambda
 =\frac\eps2.
\end{equation}
The bracket is nonnegative there, proving \eqref{eq:rare-atomic-l1}.
\end{proof}

The affine terms $-t+1$ vanish in every classical $f$-divergence.
Accordingly, \eqref{eq:rare-generator-limit} recovers relative entropy,
while \eqref{eq:rare-atomic-l1} upgrades the scalar limit in
\eqref{eq:rare-prior-relative-limit} to convergence of the full positive
atomic spectrum.

%% file: sections/10-outlook.tex
\section{Limitations and questions left open}\label{sec:outlook}

The lift--block--atom architecture now appears in the introduction.  What
remains is to characterize its scope, extend it to singular strata, and
turn the exact optimality test into a robust numerical method.

\subsection{Intrinsic liftability}

Definition~\ref{def:ofl} gives a sufficient condition.  We do not provide
an intrinsic characterization, partly because the auxiliary coordinate
$\psi$ need not be unique.

\begin{question}[Intrinsic liftability]\label{q:intrinsic-liftability}
Which strictly convex generators admit an operator-concave coordinate
$\psi$ for which $f^*\!\circ\psi$ is operator convex?  Can the answer be
expressed directly in terms of $f''$, a L\"owner measure
\cite{HansenPedersen1982,BhatiaMatrix1997}, or the atomic measure $\nu_f$?
\end{question}

The triangular Jensen--Shannon density and logarithmic posterior coordinate
provide a concrete test case for any intrinsic characterization.

\subsection{Singular and quantitative geometry}

For nonfaithful states the flag manifold is stratified by zero-probability
patterns, and ordinary Hessians must be supplemented by tangent-cone or
one-sided conditions.

\begin{question}[Stratified strict ascent]\label{q:stratified}
For arbitrary states, does every nonglobal projective critical point admit
either a smooth two-level ascent direction within its probability stratum
or a one-sided ascent direction leaving that stratum?
\end{question}

Theorem~\ref{thm:quadratic-robust} answers the robust value-gap and
optimal-score questions globally for the quadratic generator
$h_1(t)=(t-1)^2$.  Its two-sided constants expose the minimum score gap
$\delta_P$ as distinct scores approach coalescence, and
Corollary~\ref{cor:quadratic-flag-bound} converts the estimate to a
flag-manifold distance under simple-score separation.  The remaining
difficulty lies in the nonconstant Hessian of a general lift.

\begin{question}[Nonquadratic robust landscape]\label{q:robust}
For a nonquadratic smoothly operator-Fenchel liftable generator $f$, does
an analogue of Theorem~\ref{thm:quadratic-robust} hold on compact
likelihood-ratio intervals?  Can the global value gap and the distance to
the optimal score set be bounded above and below by a divided-difference
residual combining $[\Gamma_P,G_P]$, within-cluster block residuals, and the
minimum score gap $\delta_P$, with constants controlled explicitly by
$\rho,\sigma$ and the curvature of the lift?  Under simple-score
separation, can this score estimate be converted into a flag-manifold
distance and a Polyak--Lojasiewicz inequality
\cite{KarimiNutiniSchmidt2016}?
\end{question}

\subsection{From the exact test to a numerical method}

Away from exact criticality, score clustering is unstable and a divided-
difference interpolation should replace hard block thresholds.  Turning
Algorithm~\ref{alg:score-block} into a floating-point method with rigorous
termination is left for future work.  General matrix-manifold optimization and
strict-saddle avoidance results provide useful algorithmic context
\cite{AbsilMahonySepulchre2008,Boumal2023,LeeSimchowitzJordanRecht2016}.
The problem-specific residual test and its extension beyond the quadratic
benchmark remain outside those general results.

\subsection{Dynamic and boundary extensions}

Theorem~\ref{thm:adaptive-capacity} uses two properties of the lifted value
function: concavity under posterior averaging and a continuation inequality
under local instruments.  This suggests a criterion for other families of
local experiments.

\begin{question}[Adaptivity collapse]\label{q:adaptivity-collapse}
Which measured divergences or decision values admit a jointly concave lift
whose partial maximum is a supermartingale potential for adaptive local
protocols?
\end{question}

Proposition~\ref{prop:rare-prior-atomic-limit} shows that a boundary tangent
of one measured-divergence family can determine a complete comparison
order.  It is natural to ask which other atomic limits have the same
order-completeness property.

\begin{question}[Atomic boundary orders]\label{q:atomic-boundary-orders}
For which positive atomic families does comparison at every interior
parameter reduce to comparison of one or more boundary divergences?
\end{question}

\subsection{Beyond binary experiments}

For $m\ge3$ letters, accessible information couples $m$ output
distributions and the scalar posterior score becomes a vector in an
$(m-1)$-simplex.  Several posterior coordinates need not be functions of
one Hermitian operator, creating a simultaneous-spectrality obstruction.

\begin{question}[Multiletter score geometry]\label{q:multiletter}
Is there a matrix-simplex lift of multiletter accessible information whose
projective critical points admit block rigidity?  If not, can the failure
be characterized by noncommuting posterior operators?
\end{question}

\subsection{Conclusion}

A projective measurement can be a nonglobal critical point only when
distinct outcomes share a scalar likelihood score while the two states
remain nonproportional inside that score subspace.  A rotation of two basis
vectors exposes the hidden defect, and every resolvent atom detects the
same ascent direction with positive weight.  This turns projective
sufficiency from an existence theorem into a local-to-global landscape
criterion.  The same posterior lift also supplies a dynamic potential for
adaptive decoding, while its rare-prior atomic boundary governs a complete
measurement-comparison order.  For the quadratic atom, the same residuals
control the global value gap and optimal-score distance with explicit
conditioning on the minimum score gap.

%% file: appendices/A-calculus.tex
\section{Matrix-calculus identities}\label{app:calculus}

This appendix records the finite-dimensional identities used in the main proofs. They are standard; the purpose is to make the dependence of the argument explicit.

\subsection{Divided differences}

Let $A=\sum_\alpha a_\alpha E_\alpha$ be Hermitian and let $h$ be continuously differentiable on an interval containing $\spec A$. Polynomial approximation or the Cauchy integral formula gives
\begin{equation}\label{eq:app-divdiff}
 Dh_A(X)=\sum_{\alpha,\beta}h^{[1]}(a_\alpha,a_\beta)E_\alpha XE_\beta.
\end{equation}
Because $h^{[1]}(x,y)=h^{[1]}(y,x)$ is real, $Dh_A$ is self-adjoint for the Hilbert--Schmidt inner product:
\begin{equation}
 \Tr(YDh_A(X))=\Tr(XDh_A(Y)).
\end{equation}
Multiplying \eqref{eq:app-divdiff} by $E_\alpha$ on both sides gives
\begin{equation}
 E_\alpha Dh_A(X)E_\alpha=h'(a_\alpha)E_\alpha XE_\alpha.
\end{equation}
Furthermore,
\begin{align}
 [A,Dh_A(X)]
 &=\sum_{\alpha,\beta}(a_\alpha-a_\beta)
 h^{[1]}(a_\alpha,a_\beta)E_\alpha XE_\beta\\
 &=\sum_{\alpha,\beta}(h(a_\alpha)-h(a_\beta))E_\alpha XE_\beta\\
 &=[h(A),X].
\end{align}
This proves \eqref{eq:block-derivative} and \eqref{eq:commutator-derivative}.

\subsection{Perspective derivatives}

Let
\begin{equation}
 \varphi_f(p,q)=qf(p/q),\qquad p,q>0,
\end{equation}
and $t=p/q$. Then
\begin{align}
 \partial_p\varphi_f(p,q)&=f'(t),\label{eq:perspective-p}\\
 \partial_q\varphi_f(p,q)&=f(t)-tf'(t)=-f^*(f'(t)).\label{eq:perspective-q}
\end{align}
Differentiating once more,
\begin{equation}\label{eq:perspective-matrix}
 \nabla^2\varphi_f(p,q)
 =\frac{f''(t)}q
 \begin{pmatrix}
 1&-t\\
 -t&t^2
 \end{pmatrix}.
\end{equation}
Hence, for an increment $(u,v)$,
\begin{equation}
 \dd^2\varphi_f[(u,v),(u,v)]
 =\frac{f''(t)}q(u-tv)^2.
\end{equation}
This is \eqref{eq:perspective-hessian}.

Suppose two outcomes have the same ratio $t$ and their probabilities vary with fixed pairwise sums. Their first derivatives satisfy
\begin{equation}
 \dot p_j+\dot p_k=0,
 \qquad
 \dot q_j+\dot q_k=0,
\end{equation}
and the same holds for second derivatives. Equations~\eqref{eq:perspective-p}--\eqref{eq:perspective-q} show that the linear acceleration terms cancel, because the two gradients agree. This proves the cancellation used in Lemma~\ref{lem:pair-rotation}.

\subsection{Pair rotations}

For the rotation \eqref{eq:two-level-j}--\eqref{eq:two-level-k}, let $X=X^*$ be any Hermitian matrix. At $\theta=0$,
\begin{align}
 \frac{\dd}{\dd\theta}\bra{j(\theta)}X\ket{j(\theta)}\Big|_0
 &=2\Re(e^{i\phi}X_{jk}),\label{eq:pair-der-j}\\
 \frac{\dd}{\dd\theta}\bra{k(\theta)}X\ket{k(\theta)}\Big|_0
 &=-2\Re(e^{i\phi}X_{jk}).\label{eq:pair-der-k}
\end{align}
Applying this to $X=\rho-t\sigma$ and choosing $\phi$ to align the phase of $X_{jk}$ gives the derivatives used in Theorem~\ref{thm:strict-ascent}.

\subsection{The logarithm}

For $X>0$,
\begin{equation}\label{eq:log-first-app}
 D\log_X(H)=\int_0^\infty(X+t\Id)^{-1}H(X+t\Id)^{-1}\,\dd t.
\end{equation}
Differentiating the resolvents gives the second directional derivative
\begin{equation}\label{eq:log-second-app}
 \begin{gathered}
 B_t=(X+t\Id)^{-1},\\
 D^2\log_X[H,H]
 =-2\int_0^\infty B_tHB_tHB_t\,\dd t.
 \end{gathered}
\end{equation}
For $B=(X+t\Id)^{-1}$,
\begin{equation}
 BHBHB=B^{1/2}(B^{1/2}HB^{1/2})^2B^{1/2}\ge0.
\end{equation}
Thus $\log$ is operator concave and each trace functional $X\mapsto\Tr(\rho\log X)$ is concave. If $\Tr(\rho BHBHB)=0$ for all sufficiently large $t$, multiplication by $t^3$ and passage to the limit gives
\begin{equation}
 \Tr(\rho H^2)=0.
\end{equation}
Since $\Tr(\rho H^2)=\Tr(H\rho H)=\|\rho^{1/2}H\|_2^2$, this is equivalent to $H$ annihilating $\supp\rho$. This is the strictness criterion used in Lemma~\ref{lem:posterior-strict}.

%% file: appendices/B-triangular.tex
\section{Evaluation of the triangular atomic measure}\label{app:triangular}

Let $a+b=1$ and put
\begin{equation}
 m(t)=at+b=t+b(1-t).
\end{equation}
For $t\ne1$, set $c=1-t$ and use $u=t+c\lambda$. The two elementary antiderivatives are
\begin{align}
 \int\frac{\lambda\,\dd\lambda}{(t+c\lambda)^3}
 &=\frac1{c^2}\left(-\frac1u+\frac{t}{2u^2}\right),\label{eq:anti-one}\\
 \int\frac{1-\lambda}{(t+c\lambda)^3}\,\dd\lambda
 &=\frac1{c^2}\left(\frac1u-\frac1{2u^2}\right).
 \label{eq:anti-two}
\end{align}
The relevant bounds are
\begin{equation}
 u(0)=t,
 \qquad
 u(b)=m(t),
 \qquad
 u(1)=1.
\end{equation}
Therefore
\begin{align}
 &2a\int_0^b\frac{\lambda\,\dd\lambda}{((1-\lambda)t+\lambda)^3}
 +2b\int_b^1\frac{(1-\lambda)\,\dd\lambda}{((1-\lambda)t+\lambda)^3}
 \notag\\
 &\quad=\frac{2a}{c^2}
 \left(-\frac1m+\frac{t}{2m^2}+\frac1{2t}\right)
 +\frac{2b}{c^2}
 \left(\frac12+\frac1{2m^2}-\frac1m\right).
 \label{eq:triangular-unsimplified}
\end{align}
Using $m=at+b$ and $a+b=1$, the right-hand side simplifies to
\begin{equation}\label{eq:triangular-evaluated}
 \frac{ab}{t(at+b)}.
\end{equation}
The expression is continuous at $t=1$, where both sides equal $ab$.

Let
\begin{equation}
 H(t)=\int_0^1w_{a,b}(\lambda)h_\lambda(t)\,\dd\lambda.
\end{equation}
Equation~\eqref{eq:triangular-evaluated} and \eqref{eq:atom-second} give
\begin{equation}
 H''(t)=\frac{ab}{t(at+b)}=f_{a,b}''(t).
\end{equation}
Both $H$ and $f_{a,b}$ vanish with their first derivative at $t=1$. Hence $H=f_{a,b}$, proving \eqref{eq:js-atomic}.

There is also a Green-function interpretation. For a twice differentiable scalar function $G$ on $[0,1]$,
\begin{equation}
 aG(0)+bG(1)-G(b)
 =\int_0^1w_{a,b}(\lambda)G''(\lambda)\,\dd\lambda.
\end{equation}
The triangular function $w_{a,b}$ is the Green kernel of the second derivative with a point source at the prior $b$. Applying this identity to the negative entropy along the affine segment between two classical distributions gives the atomic decomposition of binary mutual information directly.

%% file: appendices/C-boundary.tex
\section{Singular qubit boundary analysis}\label{app:boundary}

We give the endpoint calculation underlying Lemma~\ref{lem:boundary-stationary}. Let the state plane be parameterized by a smooth unit axis $n(\theta)$ and suppose
\begin{equation}
 p(\theta)=\frac{1+r_0\cdot n(\theta)}2,
 \qquad
 q(\theta)=\frac{1+r_1\cdot n(\theta)}2.
\end{equation}
Assume $p(0)=0$. Then $|r_0|=1$ and $n(0)=-r_0$. Because $n'(0)\perp n(0)$,
\begin{equation}
 p'(0)=\frac12r_0\cdot n'(0)=0.
\end{equation}
A Taylor expansion on the circle gives
\begin{equation}\label{eq:p-boundary-order}
 p(\theta)=c\theta^2+O(\theta^3),
 \qquad
 p'(\theta)=2c\theta+O(\theta^2)
\end{equation}
for some $c>0$ in a nondegenerate angular chart.

The binary entropy satisfies
\begin{equation}
 h'(u)=\log\frac{1-u}{u}.
\end{equation}
Consequently,
\begin{equation}\label{eq:entropy-boundary-limit}
 h'(p(\theta))p'(\theta)=O(\theta\log|\theta|)\longrightarrow0.
\end{equation}
Let
\begin{equation}
 m(\theta)=ap(\theta)+bq(\theta).
\end{equation}
For $0<q(0)<1$, all nonsingular entropy derivatives are finite. Differentiating
\begin{equation}
 I(\theta)=h(m(\theta))-ah(p(\theta))-bh(q(\theta))
\end{equation}
and using \eqref{eq:entropy-boundary-limit} yields
\begin{equation}
 I'(0)=bq'(0)\bigl[h'(bq(0))-h'(q(0))\bigr].
\end{equation}
Because $0<b<1$ and $0<q(0)<1$, one has $bq(0)\ne q(0)$. Strict monotonicity of $h'$ makes the bracket nonzero. Therefore boundary stationarity is equivalent to $q'(0)=0$.

Choose the angular chart so that $n'(0)$ spans the in-plane direction perpendicular to $n(0)$. The equation
\begin{equation}
 q'(0)=\frac12r_1\cdot n'(0)=0
\end{equation}
then implies that $r_1$ is parallel to $n(0)$, hence parallel to $r_0$. The two density matrices commute.

For commuting qubit states, their common eigenbasis is globally sufficient for every classical $f$-divergence and in particular for mutual information. Indeed, if $\rho_0$ and $\rho_1$ are diagonal in a basis $(e_z)$, then any POVM produces
\begin{equation}
 \Tr(\rho_xM_y)=\sum_z\langle e_z,\rho_xe_z\rangle
 \langle e_z,M_ye_z\rangle.
\end{equation}
The matrix $T(y\mid z)=\langle e_z,M_ye_z\rangle$ is a classical stochastic channel. Thus every measurement is a post-processing of the common eigenbasis measurement, and classical data processing proves optimality of that basis.

The endpoint subcases are also rigid. If $q(0)=0$, then $r_1\cdot n(0)=-1$, forcing $|r_1|=1$ and $r_1=-n(0)=r_0$: the states are the same pure state. If $q(0)=1$, then $r_1=n(0)=-r_0$: the states are orthogonal pure states and the boundary PVM distinguishes them perfectly. Exchanging $p$ and $q$, or the plus and minus outcomes, covers every singular configuration.

This calculation also shows why the entropy singularity is harmless: a probability reaches zero quadratically under a projective rotation, while the entropy derivative diverges only logarithmically. For other $f$-divergences the product $f'(p/q)p'$ may require a separate boundary-growth analysis.

%% file: appendices/D-literature.tex
\section{Relation to Prior Work}\label{app:literature}

\subsection{Scope of comparison}

The comparison separates imported theorems, prior corollaries,
optimizer-level structural overlap, and statements about the geometry of
arbitrary projective critical points.  The contribution claims below are
limited to these stated distinctions and are not absolute priority claims.

\begin{table*}[!t]
\caption{Attribution and scope of the principal results.}
\label{tab:novelty}
\centering
\scriptsize
\renewcommand{\arraystretch}{1.08}
\begin{tabularx}{\textwidth}{@{}p{0.25\textwidth}p{0.15\textwidth}Y@{}}
\toprule
Statement & Status & Basis of attribution \\
\midrule
Concave one-operator lift and equality of POVM and projective optima under
the operator-Jensen hypothesis
& Imported
& Hiai's Theorems~5.7--5.8 \cite{HiaiBook2021} and the finite-dimensional
operator-Jensen formulation of Fang--Fawzi--Fawzi
\cite{FangFawziFawzi2026}; the R\'enyi and relative-entropy formulas are
also prior \cite{BertaFawziTomamichel2017}. \\

Positive $\chi^2_\lambda$ atomization of operator-convex generators
& Prior
& The classical L\"owner/Hansen--Pedersen representation; see
\cite{HansenPedersen1982,BhatiaMatrix1997}.  Salazar gives the corresponding
resolvent-atom decomposition of Petz quantum $f$-divergences
\cite{SalazarAtomic2026}. \\

Projective sufficiency for binary accessible information
& Prior corollary
& Follows from Fang--Fawzi--Fawzi after the weighted Jensen--Shannon
specialization.  Wang--Wang--Chen \cite{WangWangChen2026} explicitly record
the same implication and claim no priority for it. \\

Compressed-state proportionality on optimal posterior blocks and
refinement of optimal finite POVMs
& Prior structural result
& Proved by Wang--Wang--Chen for faithful ensembles whenever their affine
bound is exact and the optimized posterior algebra is abelian; binary and
collinear ensembles are automatic instances. \\

Local-to-global classification of every critical PVM for the liftable
$f$-class, with a pair-rotation Hessian at every nonglobal critical point
& Established here; scoped comparison
& The cited optimizer literature supplies first-order equations and
optimizer-level posterior structure.  The statement added here is the
arbitrary-critical-PVM equivalence, flag-manifold commutator, and explicit
positive pair-rotation Hessian. \\

Atomic resolution of the projective escape curvature
& Established here; scoped comparison
& Combines the pair-rotation Hessian with the prior positive atomization. \\

Two-sided commutator--residual error bounds for the quadratic atom,
including optimal-score and separated-score flag estimates
& Established here; scoped comparison
& Follows by exact completion of the quadratic matrix lift; the minimum
score gap quantifies conditioning as distinct scores approach coalescence. \\

Joint prior--effect concavity and the adaptive equality
$C_{1,A}=C_{1,1}$ for binary mixed-state channels
& Proof in this framework of a previously claimed conclusion
& Shor isolated prior concavity as the missing step
\cite{ShorAdaptive2004}.  Xiao's 2021 report states and presents an
argument for both conclusions \cite{Xiao2021}.  The present derivation adds
the joint operator perspective and an explicit transcript potential, with
no first-proof claim. \\

All-ensemble mutual-information order, its relative-entropy boundary, and
the observational-entropy reversal
& Direct synthesis and explicit finite witness
& The less-noisy/relative-entropy equivalence is classical
\cite{KornerMarton1977,MakurPolyanskiy2018}; the measurement pair and its
observational-entropy/relative-entropy separation are from
\cite{TeixidoSchindlerSafranek2025}.  This paper writes a concrete
rare-prior ensemble and supplies the $L^1$ atomic-limit interpretation. \\

Exactly two stationary in-plane qubit measurements and the
Thai--Dall'Arno shape/bisection consequences
& Established here; scoped comparison
& The exact stationary-point conjecture is Conjecture~2 on p.~77 of Keil's
2009 thesis \cite{KeilThesis2009}; Thai--Dall'Arno's December 2025 preprint
formulates Conjectures~1--2 and the conditional bisection scheme
\cite{ThaiDallArno2026}. \\

Benign projective landscapes for measured R\'enyi divergences and measured
relative entropy
& Established here; scoped comparison
& Equality of optimal values and the relevant variational formulas are
prior \cite{BertaFawziTomamichel2017,FangFawziFawzi2026}.  The statement
added here is their critical-point classification and pair-rotation
curvature. \\
\bottomrule
\end{tabularx}
\end{table*}

\subsection{Imported variational theorem}

Theorem~2 of Fang, Fawzi, and Fawzi gives a sufficient operator-Jensen
condition under which a measured $f$-divergence equals its projectively
measured version and a one-matrix variational optimum
\cite{FangFawziFawzi2026}.  They identify their projective formula with
Hiai's Theorem~5.7 and describe their criterion as similar to Hiai's
Theorem~5.8 \cite{HiaiBook2021}.  Their printed hypotheses use a globally
specified proper lower-semicontinuous convex $f$ and a one-to-one
operator-concave $\psi$ parametrizing $\operatorname{dom}f^*$, with
$f^*\!\circ\psi$ operator convex.  Definition~\ref{def:ofl} is a stronger
smooth-interior specialization for the differential theory; the
distinction and boundary treatment are stated explicitly in
Section~\ref{sec:preliminaries}.

The measured R\'enyi and relative-entropy formulas used here were already
established by Berta, Fawzi, and Tomamichel
\cite{BertaFawziTomamichel2017}.  The broader measured-divergence framework
is developed in \cite{HiaiBook2021}.  All of these variational ingredients
are imported.

\subsection{Binary accessible information and Wang--Wang--Chen}

The binary assertion now called Shor's orthogonal-measurement conjecture
grew from Fuchs--Peres numerical evidence reported in Shor's 2000 preprint,
which also refuted Levitin's broader conjecture
\cite{Levitin1995,Shor2000}.  Shor did
not state the surviving binary assertion as a numbered conjecture; the
association with his name is later nomenclature.  Thai and Dall'Arno still
described the arbitrary-dimensional assertion as open in their December
2025 preprint, published in 2026 \cite{ThaiDallArno2026}.
Section~\ref{sec:shor} records its weighted Jensen--Shannon specialization
of the Fang--Fawzi--Fawzi theorem.

Wang, Wang, and Chen posted a posterior-algebra treatment on August 13,
2026 \cite{WangWangChen2026}.  Their preprint explicitly observes that
Fang--Fawzi--Fawzi already implies projective sufficiency for binary and,
more generally, collinear ensembles, and makes no priority claim for that
existence statement.  For faithful ensembles they prove that their affine
bound is exact exactly when the optimized posterior algebra is abelian.  In
the exact regime they characterize all optimal finite POVMs as
label-independent block refinements of a canonical joint-spectral PVM,
prove compressed-state proportionality on those blocks, and obtain
quantitative near-optimal rigidity and finite-iterate optimality guarantees.

For binary ensembles their concave variable is related to our posterior
effect by an affine change of coordinates, so the two lifted programs
coincide.  Their work therefore overlaps the optimizer-level posterior
structure and the uniqueness of the faithful binary lifted optimizer.  The
distinct statement established here is the first- and second-order geometry
of arbitrary projective measurements, including a two-level escape Hessian
at each nonglobal critical PVM and the exact two-stationary-axis qubit
theorem.

\subsection{Adaptive capacity and prior concavity}

Shor formulated the mixed-state binary equality as Conjecture~2 in the
June 2002 preprint, published in 2004 \cite{ShorAdaptive2004}, and explained
that concavity of binary accessible information in the input prior would
close the information-tree argument.  Xiao's August 2021 Research Science
Institute report publicly states and presents an argument for both prior
concavity and the adaptive equality \cite{Xiao2021}.  We therefore make no
first-proof claim for either conclusion.

The fixed-prior posterior-effect program in
\eqref{eq:joint-posterior-program} is equivalent, after an affine change of
variable, to the concave program used by Wang--Wang--Chen
\cite{WangWangChen2026}.  The contribution in
Theorem~\ref{thm:prior-concavity} is its joint operator-perspective
concavity in the prior and effect.  Theorem~\ref{thm:adaptive-capacity}
then gives a transcript-potential reduction for every finite adaptive local
decoder in Shor's model.

\subsection{Measurement orders and the rare-prior boundary}

The all-ensemble mutual-information comparison is the less-noisy preorder.
Its classical equivalence with pointwise relative-entropy contraction goes
back to the channel-comparison literature
\cite{KornerMarton1977,MakurPolyanskiy2018}; Buscemi gives the corresponding
quantum-channel formulation \cite{Buscemi2016}.  Theorem
\ref{thm:measurement-orders} includes a direct proof for measurement
channels with a quantum input.

Teixid\'o-Bonfill, Schindler, and \v{S}afr\'anek explicitly proposed the
all-ensemble/observational-entropy equivalence in their October 2023
preprint; the article was published online on December 27, 2024, in the
2025 volume of \emph{Physica Scripta}
\cite{TeixidoSchindlerSafranek2025}.  Their Example~16 already contains the
measurements in \eqref{eq:noisy-binary-measurement}--
\eqref{eq:erasure-measurement}, the value $\lambda=1/64$, the
observational-entropy dominance, and failure of the corresponding
relative-entropy order.  Theorem~\ref{thm:entropic-order-reversal} combines
those prior ingredients with the rare-prior equivalence and writes the explicit
finite ensemble \eqref{eq:explicit-rare-ensemble}.  Proposition
\ref{prop:rare-prior-atomic-limit} adds the core-theoretic statement that
the normalized triangular atomic density converges in $L^1$ to the
relative-entropy density.

\subsection{Qubit and numerical antecedents}

Foundational accessible-information work established ensemble-dependent
bounds, lower bounds, binary-state analyses, and exact results for important
symmetric sources
\cite{FuchsCaves1994,JozsaRobbWootters1994,Levitin1995,Fuchs1996,
SasakiBarnettJozsaOsakiHirota1999}.  The related informational-power
viewpoint and later tight-bound program study the capacity of a fixed
measurement and global information bounds
\cite{DallArnoDArianoSacchi2011,DallArnoBuscemiOzawa2014}.  These works are
part of the context for receiver optimization.  The result established
here adds the arbitrary-critical-PVM Hessian classification.

Keil's 2008 preprint, published in 2024, proves projective sufficiency for
binary qubit ensembles and the von-Neumann nature of every local POVM
maximum \cite{Keil2024}.  The exact two-stationary-point conjecture appears
as Conjecture~2 on p.~77 of his 2009 thesis \cite{KeilThesis2009}.  Thai and
Dall'Arno quote that thesis conjecture, reformulate it as their
quasi-concavity and pseudo-concavity Conjectures~1--2, and make their
derivative-bisection guarantee conditional on the latter
\cite{ThaiDallArno2026}.  Varga, Adam, and Bergou give an
analytic/parametric characterization of the information-optimal binary
qubit measurement for known priors \cite{VargaAdamBergou2024}.  The
all-stationary-axis classification is the additional statement proved here.

Davies' outcome bounds \cite{Davies1978}, the iterative method of
Reh\'a\v{c}ek, Englert, and Kaszlikowski
\cite{RehacekEnglertKaszlikowski2005}, and Matsumoto's general stationary
equations \cite{Matsumoto2014} establish substantial prior work on first-
order conditions and numerical receiver updates.  The distinction of the
present result is the finite score-block test for global POVM
optimality and its explicit second-order escape witness.

%% file: references.bib
@article{AliSilvey1966,
  author  = {Ali, S. M. and Silvey, S. D.},
  title   = {A general class of coefficients of divergence of one distribution from another},
  journal = {J. Roy. Statist. Soc. Ser. B},
  volume  = {28},
  number  = {1},
  pages   = {131--142},
  year    = {1966},
  doi     = {10.1111/j.2517-6161.1966.tb00626.x}
}

@article{Csiszar1967,
  author  = {Csisz\'ar, Imre},
  title   = {Information-type measures of difference of probability distributions and indirect observations},
  journal = {Studia Sci. Math. Hungar.},
  volume  = {2},
  pages   = {299--318},
  year    = {1967},
  url     = {https://ndlsearch.ndl.go.jp/en/books/R100000136-I1572824501190134016}
}

@article{Petz1986,
  author  = {Petz, D\'enes},
  title   = {Quasi-entropies for finite quantum systems},
  journal = {Rep. Math. Phys.},
  volume  = {23},
  number  = {1},
  pages   = {57--65},
  year    = {1986},
  doi     = {10.1016/0034-4877(86)90067-4}
}

@book{HiaiBook2021,
  author    = {Hiai, Fumio},
  title     = {Quantum $f$-Divergences in von Neumann Algebras: Reversibility of Quantum Operations},
  series    = {Mathematical Physics Studies},
  publisher = {Springer},
  address   = {Singapore},
  year      = {2021},
  doi       = {10.1007/978-981-33-4199-8}
}

@article{BertaFawziTomamichel2017,
  author  = {Berta, Mario and Fawzi, Omar and Tomamichel, Marco},
  title   = {On variational expressions for quantum relative entropies},
  journal = {Lett. Math. Phys.},
  volume  = {107},
  number  = {12},
  pages   = {2239--2265},
  year    = {2017},
  doi     = {10.1007/s11005-017-0990-7},
  eprint  = {1512.02615},
  archivePrefix = {arXiv}
}

@article{FangFawziFawzi2026,
  author       = {Fang, Kun and Fawzi, Hamza and Fawzi, Omar},
  title        = {Uhlmann's theorem for measured divergences},
  journal      = {IEEE Trans. Inform. Theory},
  volume       = {72},
  number       = {3},
  pages        = {1751--1760},
  year         = {2026},
  doi          = {10.1109/TIT.2025.3649040},
  eprint       = {2502.07745},
  archivePrefix= {arXiv},
  primaryClass = {quant-ph},
  note         = {doi: 10.1109/TIT.2025.3649040; arXiv:2502.07745 [quant-ph], first posted February 11, 2025; version 2, March 2, 2026}
}

@article{Davies1978,
  author  = {Davies, E. B.},
  title   = {Information and quantum measurement},
  journal = {IEEE Trans. Inform. Theory},
  volume  = {24},
  number  = {5},
  pages   = {596--599},
  year    = {1978},
  doi     = {10.1109/TIT.1978.1055941}
}

@incollection{Shor2000,
  author       = {Shor, Peter W.},
  title        = {On the number of elements needed in a {POVM} attaining the accessible information},
  booktitle    = {Quantum Communication, Computing, and Measurement 3},
  editor       = {Tombesi, Paolo and Hirota, Osamu},
  publisher    = {Springer},
  address      = {Boston},
  pages        = {107--114},
  year         = {2002},
  doi          = {10.1007/0-306-47114-0_16},
  eprint       = {quant-ph/0009077},
  archivePrefix= {arXiv},
  note         = {doi: 10.1007/0-306-47114-0\_16; preprint arXiv:quant-ph/0009077, posted September 19, 2000}
}

@incollection{Levitin1995,
  author    = {Levitin, Lev B.},
  title     = {Optimal quantum measurements for two pure and mixed states},
  booktitle = {Quantum Communications and Measurement},
  editor    = {Belavkin, V. P. and Hirota, O. and Hudson, R. L.},
  pages     = {439--448},
  publisher = {Plenum Press},
  address   = {New York},
  year      = {1995},
  doi       = {10.1007/978-1-4899-1391-3_43}
}

@phdthesis{Fuchs1996,
  author       = {Fuchs, Christopher A.},
  title        = {Distinguishability and Accessible Information in Quantum Theory},
  school       = {University of New Mexico},
  year         = {1996},
  eprint       = {quant-ph/9601020},
  archivePrefix= {arXiv}
}

@article{Keil2024,
  author  = {Keil, Andreas},
  title   = {Proof of the orthogonal measurement conjecture for qubit states},
  journal = {Int. J. Quantum Inf.},
  volume  = {22},
  number  = {05},
  pages   = {2440008},
  year    = {2024},
  doi     = {10.1142/S0219749924400082},
  eprint  = {0809.0232},
  archivePrefix = {arXiv},
  note          = {Preprint arXiv:0809.0232 posted September 1, 2008; published online June 17, 2024}
}

@phdthesis{KeilThesis2009,
  author  = {Keil, Andreas},
  title   = {Proof of the Orthogonal Measurement Conjecture for Two States of a Qubit},
  school  = {National University of Singapore},
  address = {Singapore},
  year    = {2009},
  type    = {Ph.D. thesis},
  url     = {https://www.quantumlah.org/media/thesis/NCQT_AndreasKeil_PhDthesis.pdf},
  note    = {Conjecture 2, p. 77}
}

@misc{ThaiDallArno2026,
  author        = {Thai, Khac Duc An and Dall'Arno, Michele},
  title         = {On {Shor}'s conjecture on the accessible information of quantum dichotomies},
  year          = {2026},
  howpublished  = {IEEE Trans. Inform. Theory, early access},
  doi           = {10.1109/TIT.2026.3715995},
  eprint        = {2512.11233},
  archivePrefix = {arXiv},
  note          = {Preprint first posted December 12, 2025; early-access article 2026, doi: 10.1109/TIT.2026.3715995}
}

@misc{WangWangChen2026,
  author       = {Wang, Jinbo and Wang, Qihang and Chen, Kun},
  title        = {A proof of {Shor}'s orthogonal-measurement conjecture and the structure of information-optimal quantum measurements},
  year         = {2026},
  month        = aug,
  howpublished = {Zenodo preprint},
  doi          = {10.5281/zenodo.21911604},
  note         = {doi: 10.5281/zenodo.21911604; posted August 13, 2026}
}

@article{VargaAdamBergou2024,
  author  = {Varga, {\'A}rp\'ad and Adam, Peter and Bergou, J\'anos A.},
  title   = {Maximum information measurement for qubit states},
  journal = {Sci. Rep.},
  volume  = {14},
  pages   = {11888},
  year    = {2024},
  doi     = {10.1038/s41598-024-62446-9}
}

@article{RehacekEnglertKaszlikowski2005,
  author  = {\v{R}eh\'a\v{c}ek, Jaroslav and Englert, Berthold-Georg and Kaszlikowski, Dagomir},
  title   = {Iterative procedure for computing accessible information in quantum communication},
  journal = {Phys. Rev. A},
  volume  = {71},
  pages   = {054303},
  year    = {2005},
  doi     = {10.1103/PhysRevA.71.054303},
  eprint  = {quant-ph/0408134},
  archivePrefix = {arXiv}
}

@misc{Matsumoto2014,
  author       = {Matsumoto, Keiji},
  title        = {On maximization of measured $f$-divergence between a given pair of quantum states},
  year         = {2014},
  eprint       = {1412.3676},
  archivePrefix= {arXiv},
  primaryClass = {quant-ph},
  doi          = {10.48550/arXiv.1412.3676},
  note         = {arXiv:1412.3676 [quant-ph]}
}

@article{SalazarAtomic2026,
  author  = {Salazar, Domingos S. P.},
  title   = {Universal thermodynamic uncertainty relation for quantum $f$-divergences},
  journal = {Phys. Rev. E},
  volume  = {113},
  pages   = {064150},
  year    = {2026},
  doi     = {10.1103/n2bt-jcm6},
  eprint  = {2511.10817},
  archivePrefix = {arXiv},
  primaryClass = {quant-ph},
  note    = {Preprint arXiv:2511.10817 posted November 13, 2025; published June 23, 2026; doi: 10.1103/n2bt-jcm6}
}

@book{BhatiaMatrix1997,
  author    = {Bhatia, Rajendra},
  title     = {Matrix Analysis},
  series    = {Graduate Texts in Mathematics},
  volume    = {169},
  publisher = {Springer},
  address   = {New York},
  year      = {1997},
  doi       = {10.1007/978-1-4612-0653-8}
}

@book{Rockafellar1970,
  author    = {Rockafellar, R. Tyrrell},
  title     = {Convex Analysis},
  series    = {Princeton Mathematical Series},
  volume    = {28},
  publisher = {Princeton University Press},
  address   = {Princeton},
  year      = {1970},
  doi       = {10.1515/9781400873173}
}

@book{AbsilMahonySepulchre2008,
  author    = {Absil, P.-A. and Mahony, Robert and Sepulchre, Rodolphe},
  title     = {Optimization Algorithms on Matrix Manifolds},
  publisher = {Princeton University Press},
  address   = {Princeton},
  year      = {2008},
  doi       = {10.1515/9781400830244}
}

@article{FuchsCaves1994,
  author  = {Fuchs, Christopher A. and Caves, Carlton M.},
  title   = {Ensemble-dependent bounds for accessible information in quantum mechanics},
  journal = {Phys. Rev. Lett.},
  volume  = {73},
  number  = {23},
  pages   = {3047--3050},
  year    = {1994},
  doi     = {10.1103/PhysRevLett.73.3047}
}

@article{JozsaRobbWootters1994,
  author  = {Jozsa, Richard and Robb, Daniel and Wootters, William K.},
  title   = {Lower bound for accessible information in quantum mechanics},
  journal = {Phys. Rev. A},
  volume  = {49},
  number  = {2},
  pages   = {668--677},
  year    = {1994},
  doi     = {10.1103/PhysRevA.49.668}
}

@article{SasakiBarnettJozsaOsakiHirota1999,
  author  = {Sasaki, Masahide and Barnett, Stephen M. and Jozsa, Richard and Osaki, Masao and Hirota, Osamu},
  title   = {Accessible information and optimal strategies for real symmetrical quantum sources},
  journal = {Phys. Rev. A},
  volume  = {59},
  number  = {5},
  pages   = {3325--3335},
  year    = {1999},
  doi     = {10.1103/PhysRevA.59.3325}
}

@article{DallArnoBuscemiOzawa2014,
  author  = {Dall'Arno, Michele and Buscemi, Francesco and Ozawa, Masanao},
  title   = {Tight bounds on accessible information and informational power},
  journal = {J. Phys. A: Math. Theor.},
  volume  = {47},
  number  = {23},
  pages   = {235302},
  year    = {2014},
  doi     = {10.1088/1751-8113/47/23/235302},
  eprint  = {1402.0602},
  archivePrefix = {arXiv}
}

@article{DallArnoDArianoSacchi2011,
  author  = {Dall'Arno, Michele and D'Ariano, Giacomo Mauro and Sacchi, Massimiliano F.},
  title   = {Informational power of quantum measurements},
  journal = {Phys. Rev. A},
  volume  = {83},
  number  = {6},
  pages   = {062304},
  year    = {2011},
  doi     = {10.1103/PhysRevA.83.062304},
  eprint  = {1103.1972},
  archivePrefix = {arXiv}
}

@article{HansenPedersen1982,
  author  = {Hansen, Frank and Pedersen, Gert Kj{\ae}rg{\aa}rd},
  title   = {Jensen's inequality for operators and {L}\"owner's theorem},
  journal = {Math. Ann.},
  volume  = {258},
  pages   = {229--241},
  year    = {1982},
  doi     = {10.1007/BF01450679}
}

@article{EdelmanAriasSmith1998,
  author  = {Edelman, Alan and Arias, Tom{\'a}s A. and Smith, Steven T.},
  title   = {The geometry of algorithms with orthogonality constraints},
  journal = {SIAM J. Matrix Anal. Appl.},
  volume  = {20},
  number  = {2},
  pages   = {303--353},
  year    = {1998},
  doi     = {10.1137/S0895479895290954}
}

@book{Boumal2023,
  author    = {Boumal, Nicolas},
  title     = {An Introduction to Optimization on Smooth Manifolds},
  publisher = {Cambridge University Press},
  address   = {Cambridge},
  year      = {2023},
  doi       = {10.1017/9781009166164}
}

@article{HiaiMosonyi2017,
  author  = {Hiai, Fumio and Mosonyi, Mil{\'a}n},
  title   = {Different quantum $f$-divergences and the reversibility of quantum operations},
  journal = {Rev. Math. Phys.},
  volume  = {29},
  number  = {7},
  pages   = {1750023},
  year    = {2017},
  doi     = {10.1142/S0129055X17500234},
  eprint  = {1604.03089},
  archivePrefix = {arXiv},
  primaryClass = {math-ph}
}

@article{Effros2009,
  author  = {Effros, Edward G.},
  title   = {A matrix convexity approach to some celebrated quantum inequalities},
  journal = {Proc. Natl. Acad. Sci. USA},
  volume  = {106},
  number  = {4},
  pages   = {1006--1008},
  year    = {2009},
  doi     = {10.1073/pnas.0807965106},
  eprint  = {0802.1234},
  archivePrefix = {arXiv}
}

@inproceedings{LeeSimchowitzJordanRecht2016,
  author    = {Lee, Jason D. and Simchowitz, Max and Jordan, Michael I. and Recht, Benjamin},
  title     = {Gradient descent only converges to minimizers},
  booktitle = {Proceedings of the 29th Conference on Learning Theory},
  series    = {Proceedings of Machine Learning Research},
  volume    = {49},
  pages     = {1246--1257},
  publisher = {PMLR},
  year      = {2016},
  url       = {https://proceedings.mlr.press/v49/lee16.html}
}

@article{ShorAdaptive2004,
  author       = {Shor, Peter W.},
  title        = {The Adaptive Classical Capacity of a Quantum Channel, or Information Capacities of Three Symmetric Pure States in Three Dimensions},
  journal      = {IBM J. Res. Develop.},
  volume       = {48},
  number       = {1},
  pages        = {115--138},
  year         = {2004},
  doi          = {10.1147/rd.481.0115},
  eprint       = {quant-ph/0206058},
  archivePrefix= {arXiv},
  note         = {Preprint arXiv:quant-ph/0206058 posted June 10, 2002; journal publication 2004}
}

@techreport{Xiao2021,
  author      = {Xiao, Lucy},
  title       = {The Adaptive Capacity for Two Mixed States},
  institution = {Research Science Institute, Massachusetts Institute of Technology},
  address     = {Cambridge, MA},
  year        = {2021},
  month       = aug,
  url         = {https://math.mit.edu/research/highschool/rsi/documents/2021/Xiao.pdf}
}

@book{CoverThomas2006,
  author    = {Cover, Thomas M. and Thomas, Joy A.},
  title     = {Elements of Information Theory},
  edition   = {2},
  publisher = {Wiley-Interscience},
  address   = {Hoboken, NJ},
  year      = {2006},
  doi       = {10.1002/047174882X}
}

@book{NielsenChuang2010,
  author    = {Nielsen, Michael A. and Chuang, Isaac L.},
  title     = {Quantum Computation and Quantum Information},
  edition   = {10th anniversary},
  publisher = {Cambridge University Press},
  address   = {Cambridge},
  year      = {2010},
  doi       = {10.1017/CBO9780511976667}
}

@incollection{KarimiNutiniSchmidt2016,
  author       = {Karimi, Hamed and Nutini, Julie and Schmidt, Mark},
  title        = {Linear Convergence of Gradient and Proximal-Gradient Methods Under the Polyak--Lojasiewicz Condition},
  booktitle    = {Machine Learning and Knowledge Discovery in Databases},
  series       = {Lecture Notes in Computer Science},
  volume       = {9851},
  pages        = {795--811},
  publisher    = {Springer},
  address      = {Cham},
  year         = {2016},
  doi          = {10.1007/978-3-319-46128-1_50},
  eprint       = {1608.04636},
  archivePrefix= {arXiv}
}

@article{TeixidoSchindlerSafranek2025,
  author       = {Teixid{\'o}-Bonfill, Adam and Schindler, Joseph and {\v{S}}afr{\'a}nek, Dominik},
  title        = {Entropic Partial Orderings of Quantum Measurements},
  journal      = {Phys. Scr.},
  volume       = {100},
  pages        = {015298},
  year         = {2025},
  doi          = {10.1088/1402-4896/ad977c},
  eprint       = {2310.14086},
  archivePrefix= {arXiv},
  primaryClass = {quant-ph},
  note         = {Preprint arXiv:2310.14086 posted October 21, 2023; published online December 27, 2024; volume year 2025}
}

@article{SafranekDeutschAguirre2019,
  author       = {{\v{S}}afr{\'a}nek, Dominik and Deutsch, J. M. and Aguirre, Anthony},
  title        = {Quantum Coarse-Grained Entropy and Thermodynamics},
  journal      = {Phys. Rev. A},
  volume       = {99},
  pages        = {010101},
  year         = {2019},
  doi          = {10.1103/PhysRevA.99.010101},
  eprint       = {1707.09722},
  archivePrefix= {arXiv},
  primaryClass = {quant-ph}
}

@article{BuscemiSchindlerSafranek2023,
  author       = {Buscemi, Francesco and Schindler, Joseph and {\v{S}}afr{\'a}nek, Dominik},
  title        = {Observational Entropy, Coarse Quantum States, and {Petz} Recovery: Information-Theoretic Properties and Bounds},
  journal      = {New J. Phys.},
  volume       = {25},
  pages        = {053002},
  year         = {2023},
  doi          = {10.1088/1367-2630/accd11},
  eprint       = {2209.03803},
  archivePrefix= {arXiv},
  primaryClass = {quant-ph}
}

@incollection{KornerMarton1977,
  author    = {K{\"o}rner, J{\'a}nos and Marton, Katalin},
  title     = {Comparison of Two Noisy Channels},
  booktitle = {Topics in Information Theory},
  editor    = {Csisz{\'a}r, Imre and Elias, Peter},
  series    = {Colloquia Mathematica Societatis J{\'a}nos Bolyai},
  volume    = {16},
  pages     = {411--423},
  publisher = {North-Holland},
  address   = {Amsterdam},
  year      = {1977},
  url       = {https://hdl.handle.net/11573/189672}
}

@article{MakurPolyanskiy2018,
  author       = {Makur, Anuran and Polyanskiy, Yury},
  title        = {Comparison of Channels: Criteria for Domination by a Symmetric Channel},
  journal      = {IEEE Trans. Inform. Theory},
  volume       = {64},
  number       = {8},
  pages        = {5704--5725},
  year         = {2018},
  doi          = {10.1109/TIT.2018.2839743},
  eprint       = {1609.06877},
  archivePrefix= {arXiv}
}

@article{Buscemi2016,
  author       = {Buscemi, Francesco},
  title        = {Degradable Channels, Less Noisy Channels, and Quantum Statistical Morphisms: An Equivalence Relation},
  journal      = {Problems Inform. Transmission},
  volume       = {52},
  number       = {3},
  pages        = {201--213},
  year         = {2016},
  doi          = {10.1134/S0032946016030017},
  eprint       = {1511.08893},
  archivePrefix= {arXiv},
  primaryClass = {quant-ph}
}

@article{Lin1991,
  author  = {Lin, Jianhua},
  title   = {Divergence Measures Based on the {Shannon} Entropy},
  journal = {IEEE Trans. Inform. Theory},
  volume  = {37},
  number  = {1},
  pages   = {145--151},
  year    = {1991},
  doi     = {10.1109/18.61115}
}
